\documentclass{llncs}
\usepackage{makeidx}  
\usepackage{algorithm}
\usepackage{algorithmic}
\usepackage{dsfont}
\usepackage{amsmath}
\usepackage{graphicx}
\usepackage{subfigure}
\usepackage{multirow}
\usepackage{bm}

\begin{document}

\title{Budget-aware Online Task Assignment in Spatial Crowdsourcing}
\titlerunning{Online Matching Algorithm}  
%
\author{Jia-Xu Liu\inst{1,2} \and Ke Xu\inst{1} }
\authorrunning{Jia-Xu Liu et al.}   
%
%
\institute{State Key Laboratory of Software Development Environment, Beihang University, Beijing, China \and
Software College, Liaoning Technical University, Huludao, China
\email{liujiaxu@buaa.edu.cn, kexu999@gmail.com}}

\maketitle              
\begin{abstract}
The prevalence of mobile internet techniques  stimulates the emergence of various spatial crowdsourcing applications. 
Certain of the applications serve for requesters, budget providers,  who submit a batch of tasks and a fixed budget to platform with the desire to search suitable workers to complete the tasks in maximum quantity.
Platform lays stress on optimizing assignment strategies on seeking less budget-consumed worker-task pairs to meet requesters' demands. 
Existing research on the task assignment with budget constraint mostly focuses on static offline scenarios, where the spatiotemporal information of all workers and tasks is known in advance. 
However, workers usually appear dynamically on real spatial crowdsourcing platforms, where existing solutions can hardly handle it. 
In this paper, we formally define a novel problem \underline{B}udget-aware \underline{O}nline task \underline{A}ssignment(BOA) in spatial crowdsourcing applications. 
BOA aims to maximize the number of assigned worker-task pairs under a budget constraint where workers appear dynamically on platforms. 
To address the BOA problem, we first propose an efficient threshold-based greedy algorithm Greedy-RT which utilizes a random generated threshold to prune the pairs with large travel cost. 
Greedy-RT performs well in adversary model when compared with simple greedy algorithm, but it is unstable in random model for its randomly generated threshold may produce poor quality in matching size. 
We then propose a revised algorithm Greedy-OT which could learn approximately optimal threshold from historical data, and consequently improves matching size significantly in both models. 
Finally, we verify the effectiveness and efficiency of the proposed methods through extensive experiments on real and synthetic datasets.
\end{abstract}

\section{Introduction}
The pervasion of mobile phones has caused the rapid emergence of various of spatial crowdsourcing applications such as Uber, Gigwalk, and Waze etc. in the last decade.
Part of these applications feature with budget constraint and usually compatible with the occasions like data acquisition\cite{To2016Real} and sampling survey\cite{Restuccia2014FIDES}.
In these real applications, a batch of tasks are released by a single requester who as well provides a fixed budget for rewarding workers. Requester expects that platform could expend the fund efficiently so that more tasks are able to be assigned even if the spatiotemporal information of workers is unknown for platform util they appear.
That is, platform is required to perform online task assignment to maximize the number of total assigned pairs under a fixed budget.

Previous studies on the problem of online task assignment can generally be divided into two categories: non-guidance and guidance.
The majority studies lay in the scope of non-guidance matching. 
Multiple-armed bandit\cite{To2016Real,Hassan2014A} and random-threshold greedy\cite{Tong2016Onlinemobil,Ting2015Near} are prevalently adopted to solve the problem. 
However, neither of them could approach optimal solution since the global spatiotemporal information cannot be acquired in advance to guide online matching. 
Tong et al. proposed a guided matching solution for the problem with prior knowledge which could be estimated from the historical traces of workers and tasks\cite{Tong2017Flexible}. 
Higher performance may be achieved in guidance on the condition that the real spatial-temporal distribution of works and tasks is highly similar with their historical distribution.
Once the condition is broken when emergent events occurs, a great deviation will be produced between historical and real data. 
As a result, much more inferior matching results are generated and even worse than the non-guidance matching models in performance.
Unfortunately, the deviation between real-time and historical distribution in the real world is frequently nonnegligible over a short time, even if the spatiotemporal distribution of workers keeps periodic similarity in the magnitude of long term. 
Thus, many unsuitable matches will emerge if fine-grained historical information is directly utilized to guid real-time task assignment.
Besides the hardness in usage of historical data, another obstacle is that existing online assignment algorithms have defects in improving matching size.
For example, greedy is a simple and efficient algorithm in solving online task assignment problem\cite{Tong2016Online}. However, its performance is sensitive to the arrival order of workers, especially fragile in adversary model.  We go through the following toy example to illustrate it.
\begin{figure}
\centering
\subfigure[Initial locations]{\includegraphics[width=0.3\textwidth]{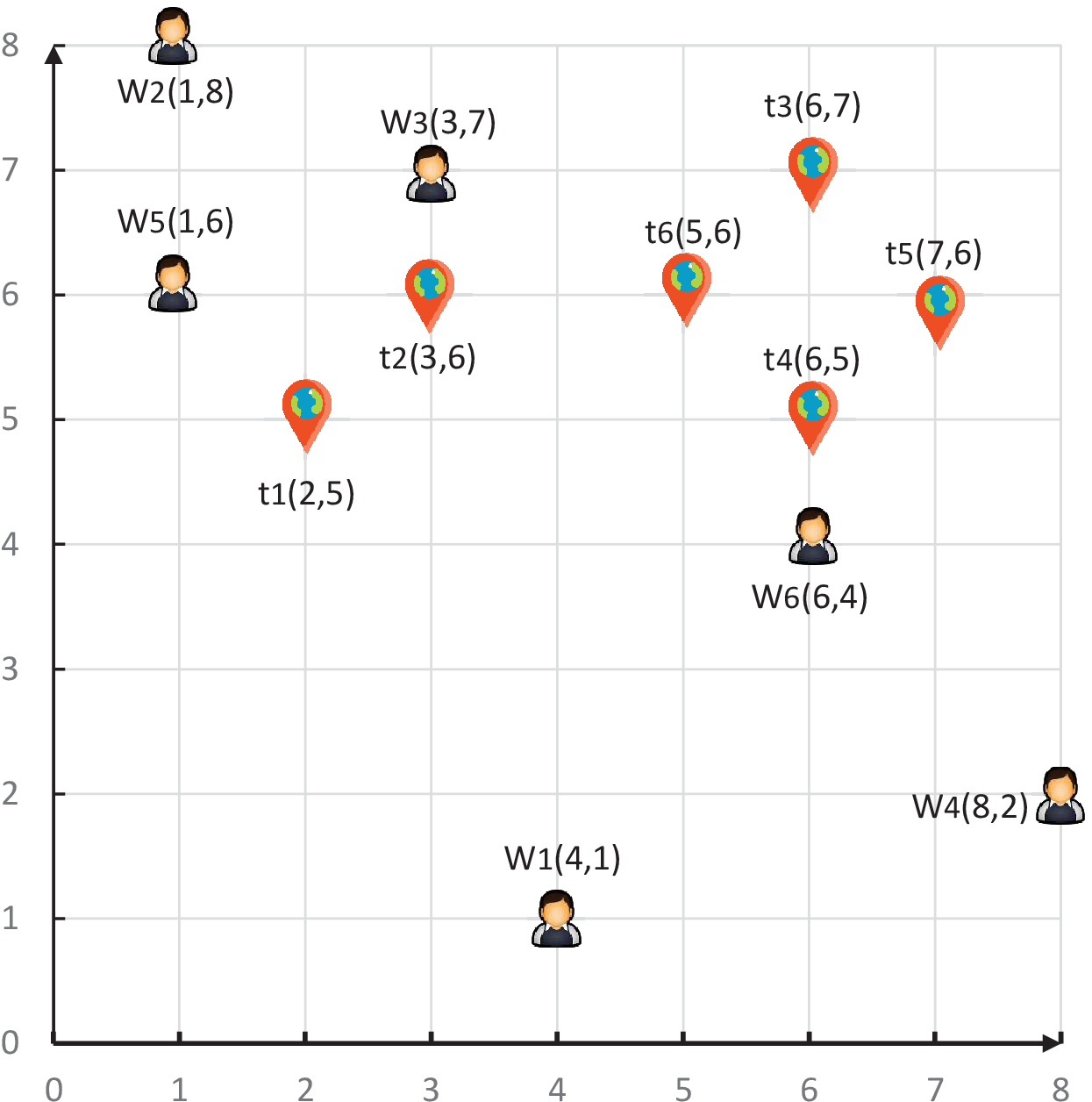}\label{fg:Init}}
\subfigure[OPT result]{\includegraphics[width=0.3\textwidth]{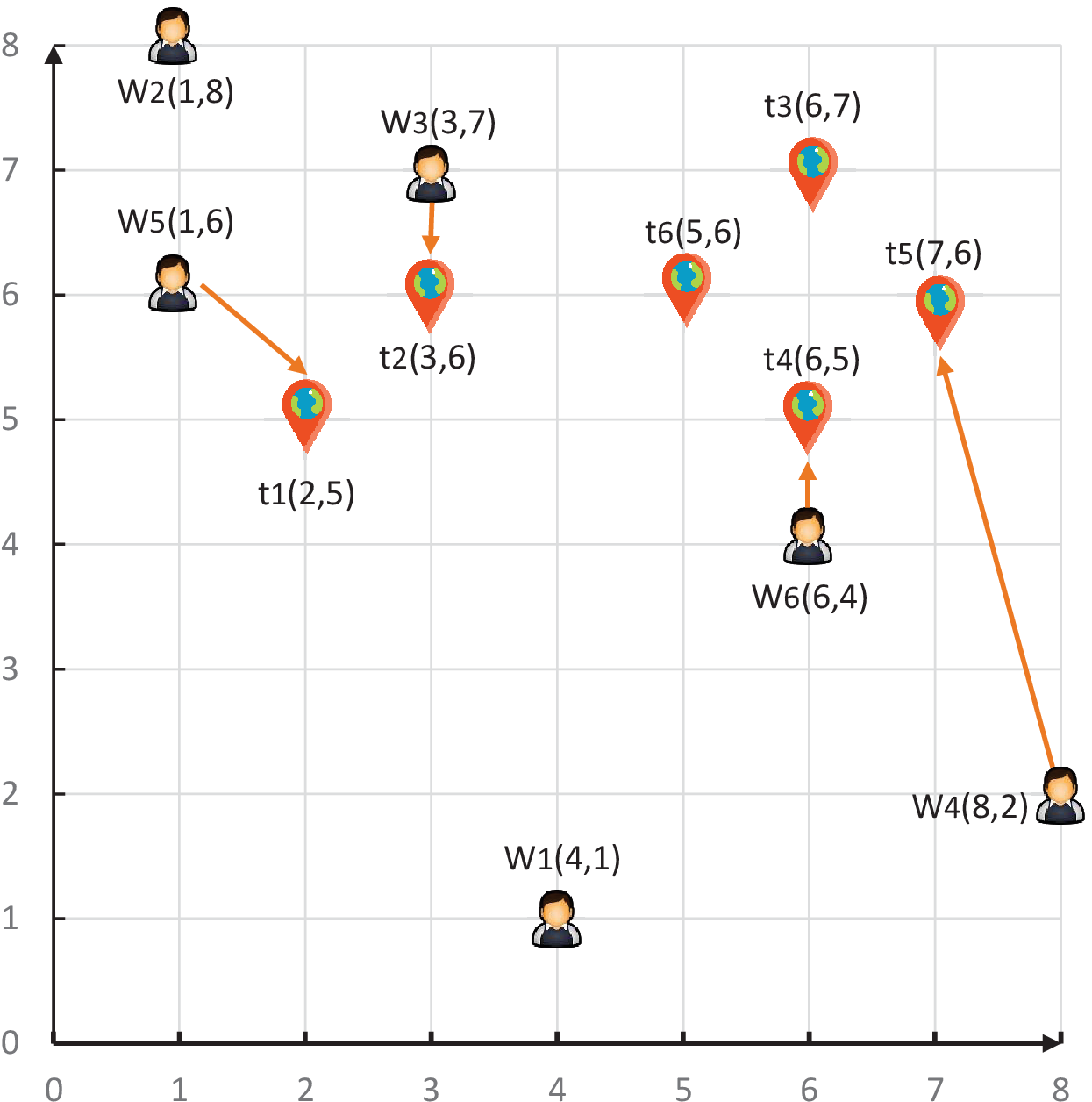}\label{fg:OPT}}
\subfigure[Greedy result]{\includegraphics[width=0.3\textwidth]{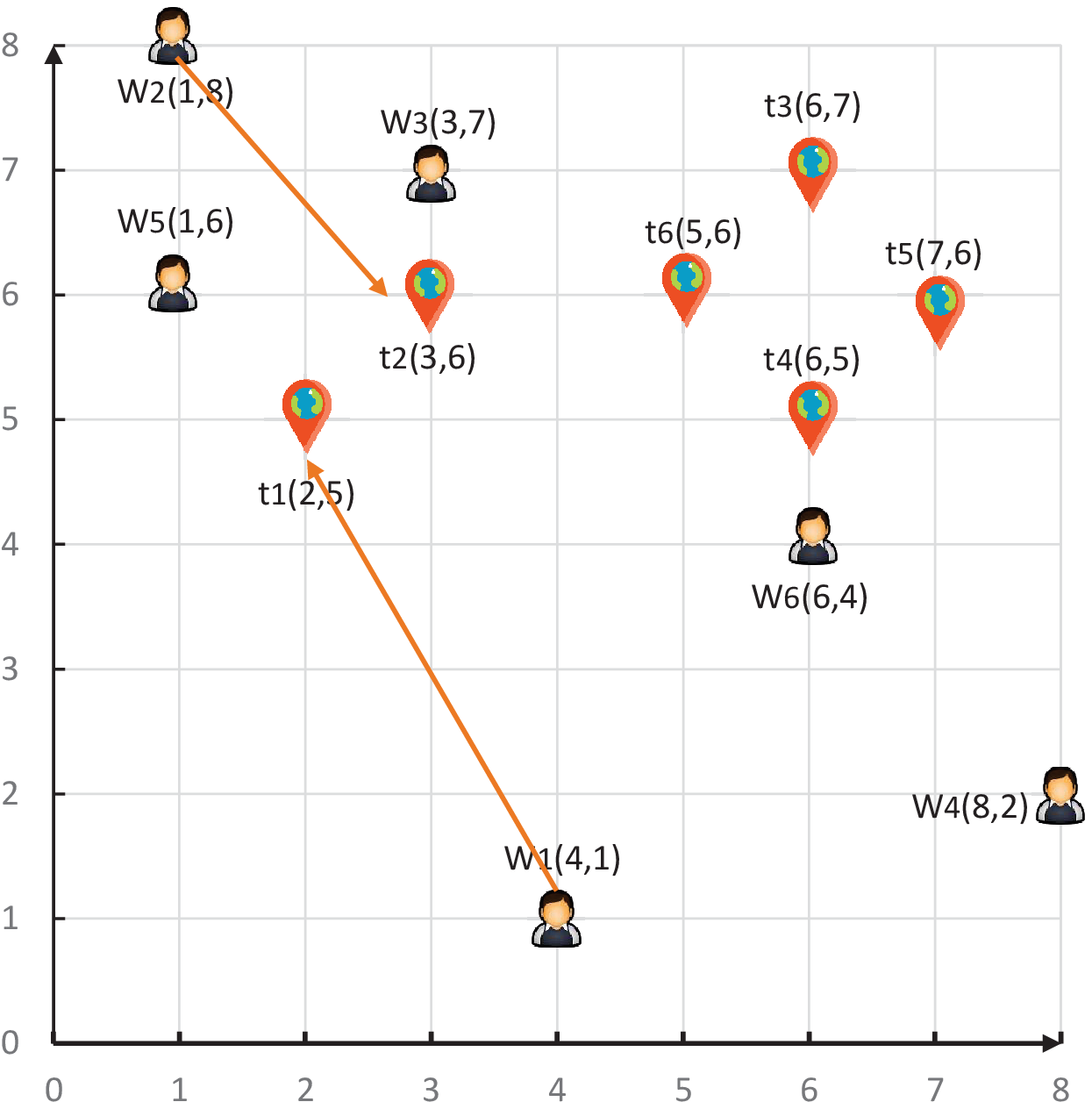}\label{fg:Greedy}}
\caption{The matching results of offline optimal algorithm and simple greedy algorithm }
\label{fg:example}
\end{figure}
\begin{example}\label{ex:1}
David submits six tasks to platform with the desire that platform helps him assign as many workers as possible to his tasks while the payment does not exceed the budget 10. 
All the tasks share the same release time 0 and expired time 10 while the appearance of workers may differ in time. 
Here, we deploy that any worker $w_i$ appears at time $i$.
Platform performs one-to-one matching between workers and tasks, where each task is completed by one worker and each worker is assigned to one task at most.
The initial locations of the tasks and workers are labeled in the 2D space($X$,$Y$) in Fig. \ref{fg:Init}.
We assume the reward for workers equals to their travel cost, then the total travel cost of assigned workers can not beyond 10.
The optimal matching is shown in Fig. \ref{fg:OPT} while the result of simple greedy is shown in Fig. \ref{fg:Greedy}.
Here, the solution of simple greedy is $\{(w_1,t_1),(w_2,t_2)\}$ whose matching size is merely 2 and used budget is 10. 
There are two reasons for the unsatisfactory result: 1) bad arrival order of workers, which causes the budget is exhausted by those early arrival workers with large travel cost. 2) simple greedy has no prior knowledge on the optimal matching which can be estimated from historical data and be used to guid the real-time task assignment.
\end{example}
Motivated by the above example, we first formulate a novel problem called Budget-aware Online task Assignment(BOA). 
To address the BOA problem, we propose two greedy variants to amend the two defects of the simple greedy mentioned above. 
In order to economize budget, the first variant utilizes a random generated threshold to filter out those serious budget-expended matchings. 
Random thresholds cannot guarantee the quality of results since small thresholds maybe filter out those eligible matchings. 
Another variant ameliorates the former by learning the near optimal threshold from historical data.
The main contribution of this work are summarized as follows.
\begin{itemize}
\item Inspired by certain emerging spatial crowdsourcing applications, we propose and formulate the BOA problem which is a problem of online task assignment in real-time spatial data with a budget constraint. 
\item To address the BOA problem, we propose a greedy variant, called Greedy-RT, which filters out those bad matchings with large expensive cost by a random generated threshold, whose performance is superior than simple greedy algorithm in adversary model.
\item We then propose another variant, called Greedy-OT, which extracts near optimal threshold from the offline optimal solution based on the spatiotemporal information of historical workers and released tasks. Greedy-OT improves matching size significantly than Greedy-RT and simple greedy algorithm both in random model and in adversary model.
\item We verify the effectiveness and efficiency of our algorithms on both synthetic dataset and large-scale Uber pickups dataset.
\end{itemize}
In the rest of this paper, we formulate the BOA problem in Section 2 and provide its offline optimal solution in Section 3.1. We present a random generated threshold greedy algorithm and analyze its competitive ratio in Section 3.2, and then present a revised greedy algorithm based on near optimal threshold and analyze its competitive ratio in Section 3.3. Section 4 presents the performance evaluations, Section 5 reviews related worker and Section 6 concludes this paper.
\section{Problem Statement}
In this section, we first introduce the basic concepts, and then formally define the Budget-aware Online task Assignment (BOA) problem.
\begin{definition}[Worker]
A worker is denoted by $w=<\bm{l}_w, b_w, e_w, v_w>$. $\bm{l}_w$ is the location of $w$ in a 2D space. $b_w$ and $e_w$ is the arrival time and leaving time on platform. $v_w$ is the velocity of $w$. 
\end{definition}
It is notable that $b_w$ equals with $e_w$ in our BOA problem, that is, when a new worker $w$ arrives, platform will instantly make an assignment determination for $w$. Once $w$ fails in matching, platform will not take him into consideration in the subsequent assignments.
\begin{definition}[Task]
A task is denoted by $t=<\bm{l}_t, r_t, d_t>$. $\bm{l}_t$ is the location of $t$ in a 2D space. $r_w$ and $d_t$ is the release time and deadline of $t$. 
\end{definition}
\begin{definition}[Task requester]
A task requester is denoted by $r=<T, B>$. ${T}=\{t_1,...,t_n\}$ is a batch of tasks released by a single requester. $B$ is the budget that requester supplies to reward workers who can complete any task of $T$. 
\end{definition}
\begin{definition}[Travel Cost]
The travel cost, denoted by $cost(w, t)$, is the distance cost for $w$ to arrive at the location of $t$, that is, the distance between $\bm{l}_w$ and $\bm{l}_t$.
\end{definition}
\begin{definition}[BOA Problem]
Given a task requester $r$, and a set of workers $W$ where workers dynamically appear on platform. 
The goal of BOA is to find an assignment scheme $M$ between $W$ and $r.T$ to maximize the number of assigned pairs.
That is, $\max Sum(M)=\sum_{w\in W, t\in r.T}I(w, t)$, where $I(w, t)=1$ if the pair $(w, t)$ is matched in the assignment $M$, and otherwise $I(w, t) = 0$, such that the following constraints are satisfied.
\begin{itemize}
\item Budget constraint. 
We assume the reward that platform supplies for a worker is proportion with his travel cost at the ratio of 1,
then the total rewards paid for workers in $M$ is less than $B$. (i.e.,$\sum_{(w,t)\in M}{cost(w,t)}\le B$).
\item Deadline constraint. 
For any worker-task pair $(w, t)$, $w$ should be able to arrive at the location of $t$ before its deadline. (i.e.,$b_w+cost(w,t)/v_w\le d_t$).
\item Invariable constraint. 
Once a task $t$ is assigned to a worker $w$, the assignment of $(w, t)$ cannot be revoked.
\end{itemize}
\end{definition}
\section{Solutions for BOA}
\subsection{Offline Optimal Solution}
In this section, we introduce the optimal solution for BOA problem, which can be solved in offline scenarios, 
where platform has acquired the entire and determined information of workers and tasks about their locations and deadline. 
\begin{theorem}\label{th1}
BOA problem is reducible to the minimum-cost maximum-flow problem.
\end{theorem}
\begin{proof}
Given $W=\{w_1, w_2, ...\}$ as the set of online workers, and $T =\{t_1, t_2,...\}$as the set of tasks to be assigned. 
Let $G =(V, E)$ be the flow network graph with $V$ as the set of vertices, and $E$ as the set of edges. 
The set $V$ contains $|W|+|T|+2$ vertices, that comprises worker vertices denoted by $V_{w_i},i\in[1,|W|]$, task vertices marked by $V_{t_j},j\in[0,|T|]$ and two additional vertices, source vertex $V_s$ and destination vertex $V_e$ respectively. 
The set $E$ contains $|W|+|T|+|W|\cdot|T|$ edges, each of which has a capacity of 1.
We associate the cost of edge from vertex $V_{w_i}$ to $V_{t_j}$ with their travel distance, and the one of other edges with 0. 
Thus, given a fixed budget, we may obtain the maximum matching in quantity by running the min-cost max-flow algorithm in $G$.   
\end{proof}
Next we will describe the steps to obtain the optimal solution of BOA problem. 
According to Theorem \ref{th1}, we first utilize the offline spatial-temporal information of workers and tasks to build bipartition graph, and then apply traditional min-cost max-flow algorithm to get the optimal matching $\hat{M}^*$.
Finally, we sort the matching pairs in $\hat{M}^*$ in ascending order, and then choose them one by one until the total cost beyonds the budget. 
The whole procedure is depicted in Algorithm \ref{al:opt}.
\begin{algorithm}[!h]
\caption{Optimal Solution}\label{al:opt}
\begin{algorithmic}[1]
\REQUIRE $W$,$T$,$B$
\ENSURE matching scheme $\hat{M}^*$
\STATE $\hat{M}^*\gets \emptyset$, $c\gets0$
\STATE create source vertex $s$, sink vertex $e$
\FORALL {worker node $w\in W$}
\STATE add\_edge($s$,$w$,1,0)
\ENDFOR
\FORALL {task node $t\in T$}
\STATE add\_edge($t$,$e$,1,0)
\ENDFOR
\FORALL {task node $w\in W$}
  \FORALL {task node $t\in T$}
    \IF{$b_w+cost(w,t)/v_w\le d_t$}
      \STATE add\_edge($w$,$t$,1,$cost(w, t)$);
    \ENDIF
  \ENDFOR
\ENDFOR
\STATE $\hat{M}\gets$ Min-Cost\_Max-Flow($s$,$e$)
\STATE sort worker-task pairs in $\hat{M}$ by cost in ascending order
\FORALL {sorted worker-task pair $(w,t) \in \hat{M}$}
  \IF{$c+cost(w,t)\le B$}
    \STATE insert $(w,t)$ into $M^*$
    \STATE $c\gets c+cost(w,t)$
  \ENDIF
\ENDFOR
\RETURN $\hat{M}^*$
\end{algorithmic}
\end{algorithm}
\begin{figure}
\centering
\subfigure[Build graph]{\includegraphics[width=0.3\textwidth]{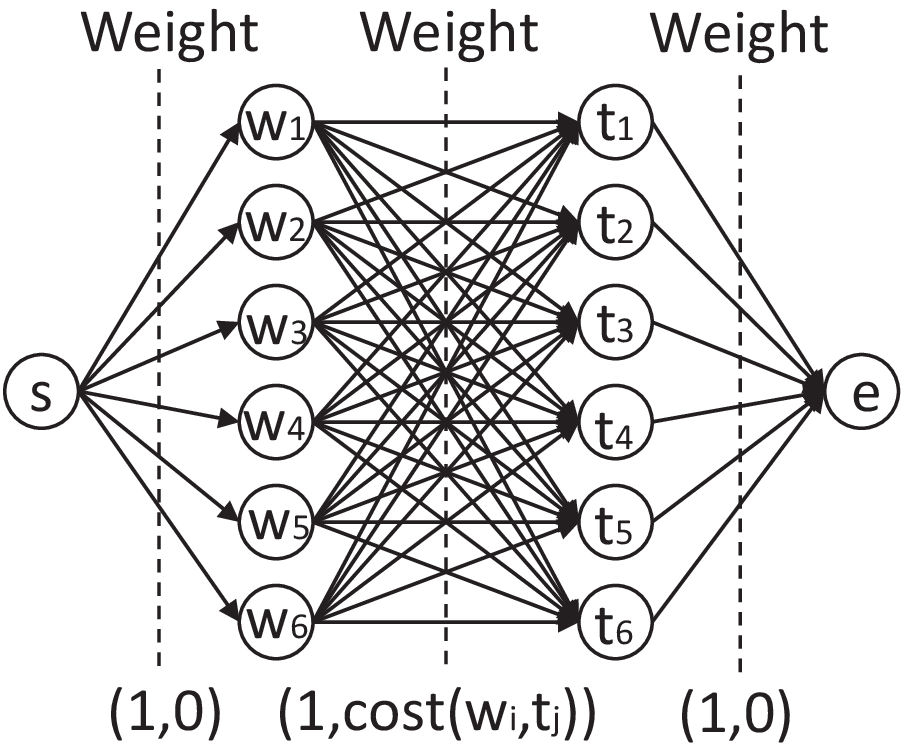}\label{fg:opt_graph1}}
\subfigure[Min-cost max-flow]{\includegraphics[width=0.3\textwidth]{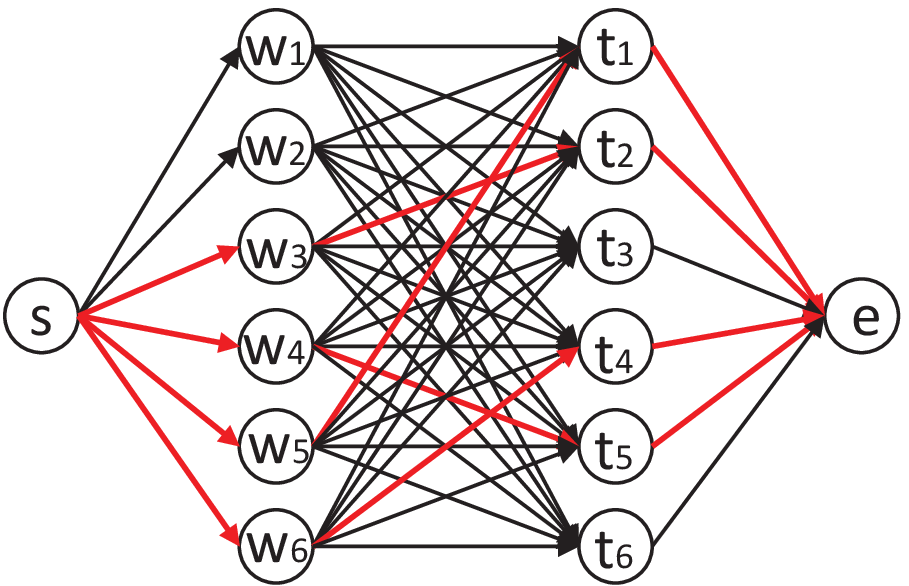}\label{fg:opt_graph2}}
\subfigure[Matching pairs]{\includegraphics[width=0.3\textwidth]{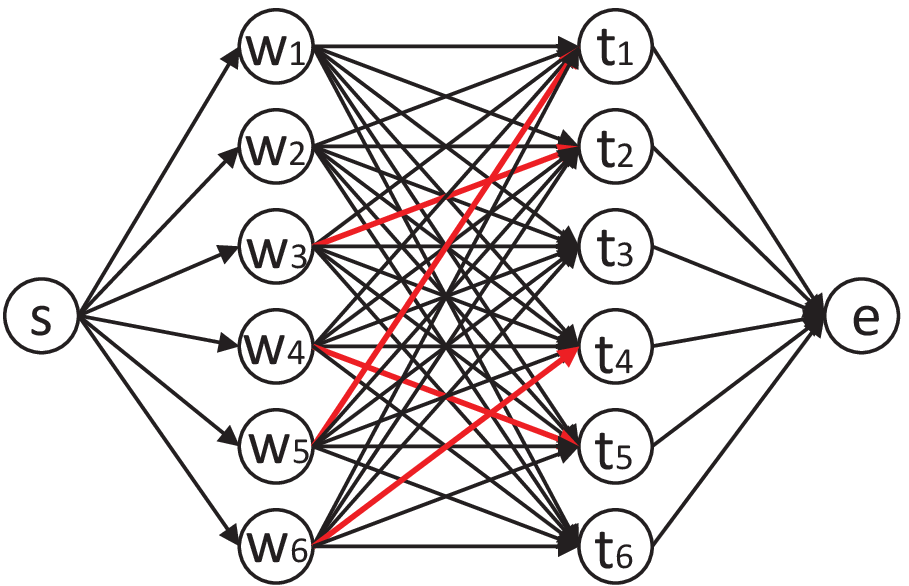}\label{fg:opt_graph3}}
\caption{Illustrated example of OPT algorithm}
\label{fg:opt}
\end{figure}
\begin{example}
Backing to our running example in Example \ref{ex:1}, As shown in Fig. \ref{fg:opt}, four flows are picked out by running the min-cost max-flow algorithm, which are\\
$f_1: s->w_3->t_2->e$ with cost 1,\\
$f_2: s->w_4->t_5->e$ with cost 5,\\
$f_3: s->w_5->t_1->e$ with cost 2,\\
$f_4: s->w_6->t_4->e$ with cost 1.\\
The total used budget is 9. Worker and task vertices are eventually extracted out from each flow, and then form the final pair set $\hat{M}^*=\{(w_3,t_2), (w_4,t_5), (w_5,t_1), (w_6,t_4)\}$.
\end{example}
\subsection{Greedy-RT Algorithm}
Greedy is a simple and efficient method for most online task matching problems, but its performance is susceptible to the order of workers' appearance. 
Its competitive ratio achieves the worst $\frac{1}{2^{min\{|T|,|W|\}}-1}$ \cite{Tong2016Online} when workers' appearance follows adversarial model. 
In order to alleviate the impaction of the order, random-threshold greedy(Greedy-RT)\cite{Ting2015Near} is competent than simple greedy methods for its threshold, which is randomly generated, could feasibly filter out those extremely bad matching pairs. 
Specifically, Greedy-RT first produces a random threshold $\tau$, and then the pairs whose travel cost is greater than $\tau$ are denied, so that the abused budget which caused by early workers in adversary model is restrained. 
In this section, we utilize the random-threshold greedy to solve BOA problem. 
\begin{algorithm}[!h]
\caption{Greedy-RT Algorithm}
\label{al:greedy-rt}
\begin{algorithmic}[1]
\REQUIRE $B$, $c_{max}$
\ENSURE matching pattern $M$
\STATE $M\gets \emptyset$,$c\gets 0$
\STATE choose $\kappa$ randomly from the set \{0,1,...,$\lceil ln(c_{max}+1)\rceil$\} with the probability $Pr(\kappa=i)=\frac{1}{|\lceil ln(c_{max}+1)\rceil|}$
\STATE $\tau\gets e^\kappa$
\FORALL {new arrival worker $w$}
    \STATE $T^\prime\gets$\{$\forall t\vert b_w+cost(w,t)/v_w\le d_t\land cost(w,t)\le \tau\land c+cost(w,t)\le B$ \}
    \IF{$T^\prime\neq \emptyset$}
        \STATE $t=\min\limits_{t\in T^\prime}cost(w,t)$
        \STATE $M\gets M\cup$($w$,$t$)
        \STATE $c\gets c+cost(w,t)$
    \ENDIF
\ENDFOR
\RETURN $M$
\end{algorithmic}
\end{algorithm}

The whole procedure of Greedy-RT is illustrated in Algorithm \ref{al:greedy-rt}.
 In line 1, the result set $M$ and the used budget $c$ are assigned with initial values. 
 In lines 2-3, Greedy-RT randomly chooses a threshold $e^\kappa$ on travel cost according to the estimated maximum cost $c_{max}$, which can be learned from the scope that workers and tasks appear.
 In lines 4-5, when a worker $w$ arrives, Greedy-RT filters a task subset $T^\prime$ where each task in $T^\prime$ satisfies three conditions:1)$w$ could arrive before $t$'s deadline; 2)the travel cost between $w$ and $t$ is little than the specific threshold $e^\kappa$; 3)the used budget $c$ after adding the cost $d(w,t)$ cannot beyond the total budget $B$. 
 In lines 6-9, if $T^\prime$ is not empty, Greedy-RT chooses the nearest task from $T^\prime$, adds it into $M$ and updates the used budget $c$.
 In line 10, the algorithm returns the final matching scheme when all the workers have already appeared or budget has been used up.
\begin{example}
Backing to our running example in Example \ref{ex:1}, workers and tasks appear in the 8$\times$8 square area where the maximum travel cost does not beyond the manhattan distance 16. 
According to the upper bound $\lceil ln(c_{max}+1)\rceil$, Greedy-RT divides the whole travel cost into four grades whose associated thresholds are $e^0$, $e^1$, $e^2$ and $e^3$ respectively. 
The matching schemes based on various thresholds are listed as follows:\\
$e^0$:$\{(w_3,t_2),(w_6,t_4)\}$, matching size:2, used budget:2\\
$e^1$:$\{(w_3,t_2),(w_5,t_1),(w_6,t_4)\}$, matching size:3, used budget:4\\
$e^2$:$\{(w_1,t_1),(w_2,t_2)\}$, matching size:2, used budget:10\\
$e^3$:$\{(w_1,t_1),(w_2,t_2)\}$, matching size:2, used budget:10.\\
Greedy-RT randomly chooses one threshold, so its expectation of matching size is $E(Greedy-RT)=\frac{2}{4}+\frac{3}{4}+\frac{2}{4}+\frac{2}{4}=2.5$ in this example, which outperforms than the simple greedy algorithm.
\end{example}
We next analyze the competitive ratio of Greedy-RT. Let $\mathcal{O}$ be the optimal matching with the limited budget $B$. 
Define $\mathcal{O}(e^i, e^{i+1}]=\{(w,t)\in \mathcal{O}|cost(w,t)\in (e^i, e^{i+1}]\}$ to be the subset of pairs in $\mathcal{O}$ whose cost are in the interval ($e^i$,$e^{i+1}$]. For any $i \ge 0$, let $M_{\le e^i}$ denote the matching scheme returned by Greedy-RT whose threshold $\tau$ equals $e^i$, or equivalently, when $\kappa = i$. 
\begin{lemma}\label{lm1}
For any $i\ge 0$, 
$|M_{\le e^i}|\ge 
\begin{cases}
|\mathcal{O}(e^{i-1},e^i]|&i\ge1 \\
|\mathcal{O}(0,e^0]|&i=0
\end{cases}
$.
\end{lemma}
\begin{proof}
when $i=0$, matching pattern $M_{\le e^0}$ is achieved with the budget $B$ while $\mathcal{O}(0,e^0]$ is done with budget $B_{(0,e^0]}$. Since $B_{(0,e^0]}\le B$, then $|M_{\le e^0}|\ge|\mathcal{O}(0,e^0]|$.
when $i\ge1$,  we have $|M_{\le e^i}|=|M_{(0,e^0]}^{B_0}|+\sum_{j=1}^{i}|M_{(e^{j-1},e^j]}^{B_j}|$, where $\sum_{j=0}^{i}B_j=B$ and $B_j$  is determined by the Greedy-RT algorithm. 
Since $|M_{(0,e^0]}^{B_0}|\ge |M_{(e^{i-1},e^{i}]}^{B_0}|$ and $|M_{(e^{j-1},e^j]}^{B_j}|\ge|M_{(e^{i-1},e^{i}]}^{B_j}|$ for any $j\in[1,i]$, we have
\begin{equation}
|M_{\le e^i}|\ge|M_{(e^{i-1},e^{i}]}^{B_0}|+\sum_{j=1}^{i}|M_{(e^{i-1},e^{i}]}^{B_j}|=|M_{(e^{i-1},e^{i}]}^{B}|
\end{equation}
Similarly, $\mathcal{O}(e^{i-1},e^{i}]$ is achieved with the budget $B_{(e^{i-1},e^{i}]}$ and $B_{(e^{i-1},e^{i}]}\le B$, we have
\begin{equation}
|M_{\le e^i}|\ge|\mathcal{O}(e^{i-1},e^{i}]|.
\end{equation}
The lemma is proved. 
\end{proof}
\begin{theorem}\label{th2}
The competitive ratio of Greedy-RT Algorithm is not less than $\frac{1}{\lceil ln(c_{max}+1)\rceil+1}$.
\end{theorem}
\begin{proof}
Let $n=\lceil ln(c_{max}+1)\rceil$. Since the exponential part of threshold is chosen evenly from a set of integers between $0$ and $n$, we have 
\begin{equation}
E(|M|)=\sum_{i=0}^{n}|M_{\le e^i}|p_i=\frac{1}{n+1}\sum_{i=0}^n|M_{\le e^i}|.
\end{equation}
According to Lemma \ref{lm1}, 
\begin{equation}
E(|M|)\ge\frac{1}{n+1}(|\mathcal{O}(0,e^0]|+|\sum_{i=1}^n|\mathcal{O}(e^{i-1},e^i]|)=\frac{1}{n+1}|\mathcal{O}|.
\end{equation}
That is, 
\begin{equation}
\frac{E(M)}{|\mathcal{O}|}\ge\frac{1}{\lceil ln(c_{max}+1)\rceil+1}.
\end{equation}
The theorem follows.
\end{proof}
\subsection{Greedy-OT Algorithm}
Greedy-RT is unstable due to its randomly selected threshold. 
Specifically, certain rational pairs will be neglected if the chosen threshold is too small, which leads its performance turns even worsen than the simple greedy. 
It is crucial for threshold-based greedy algorithms to choose an appropriate threshold. 
In this section, we propose a superior greedy variant, called Greedy-OT, which can generate a fixed near-optimal threshold.

There are periodic similarities on human labor and travel traces. For an instance, we arbitrarily extract six-days samples from Uber dataset in May 2014, each of which comprises the detail pickup records happened in Manhattan, New York city.
Fig. \ref{fg:distribution} depicts the distribution of pickups happened between 0 o'clock and 12 o'clock with heat mapps. 
\begin{figure}
\centering
\subfigure[May 6(Tue)]{\includegraphics[width=0.3\textwidth,height=3cm]{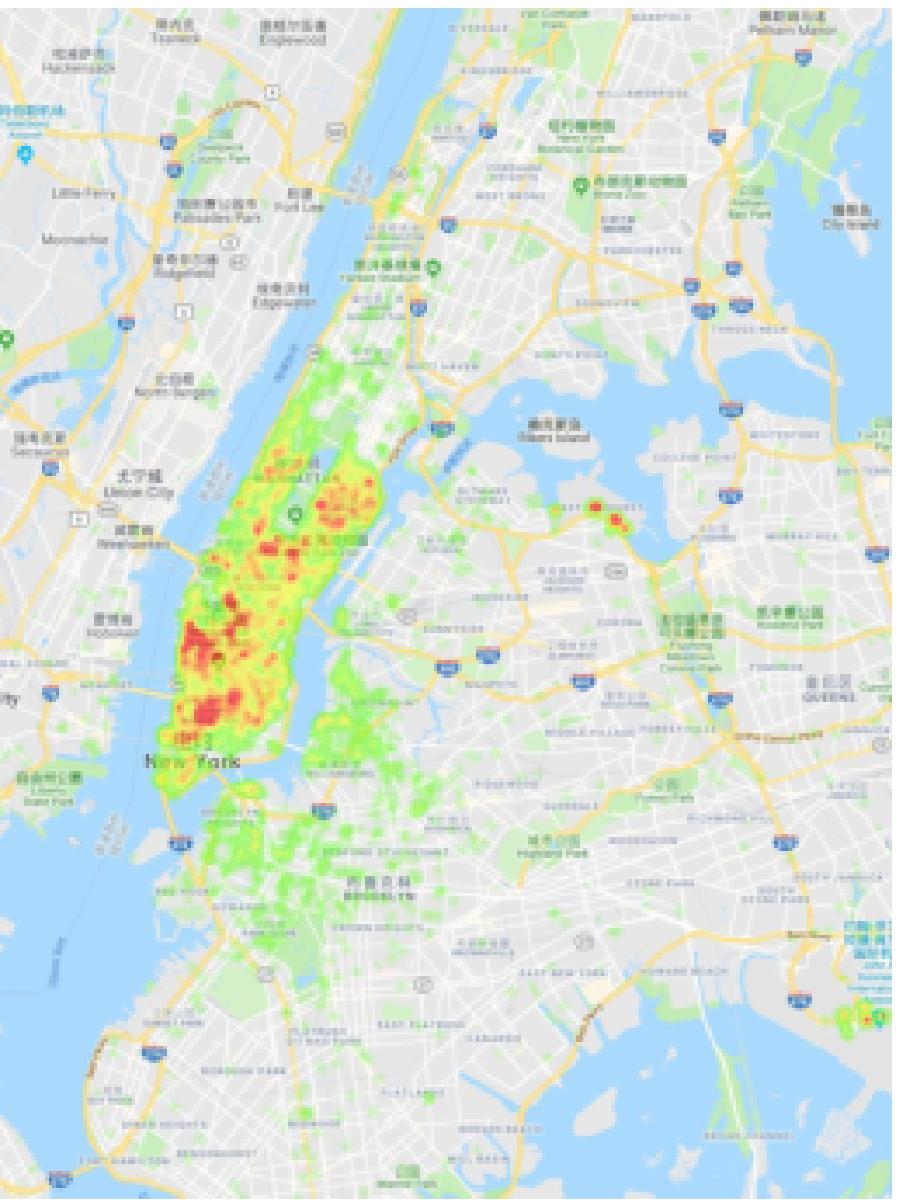}\label{fg:5-6}}
\subfigure[May 7(Wed)]{\includegraphics[width=0.3\textwidth,height=3cm]{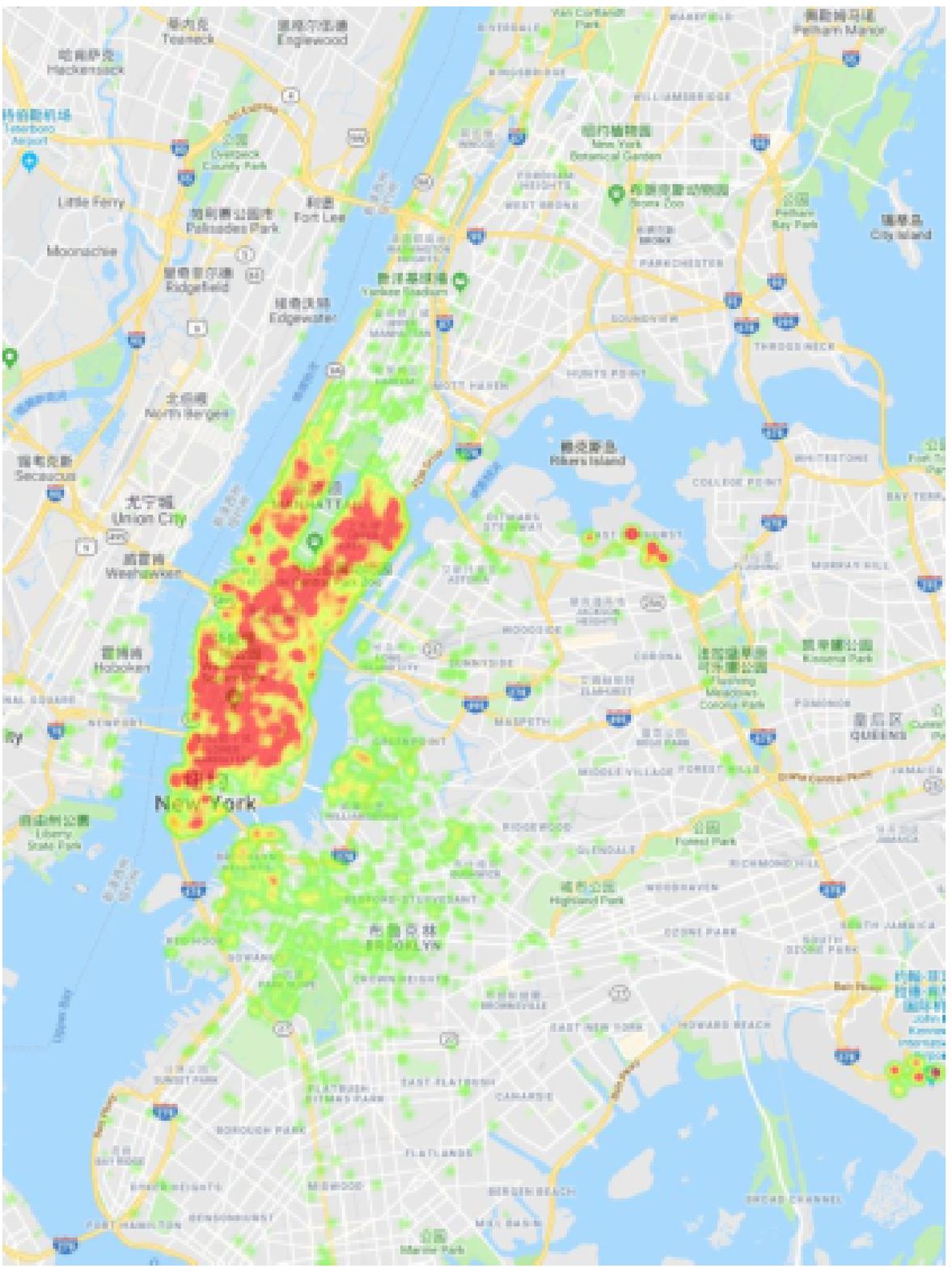}\label{fg:5-7}}
\subfigure[May 8(Thu)]{\includegraphics[width=0.3\textwidth,height=3cm]{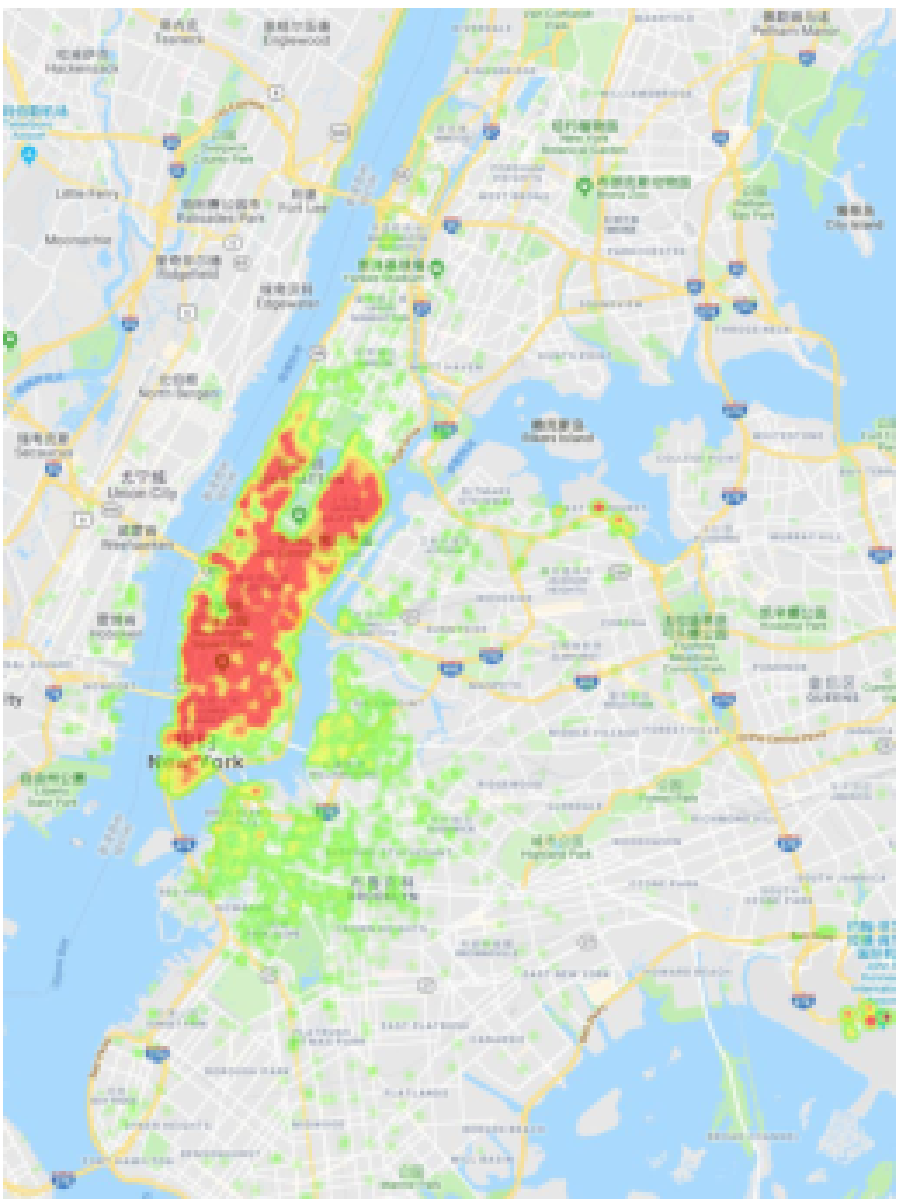}\label{fg:5-8}}
\vfill
\subfigure[May 10(Sat)]{\includegraphics[width=0.3\textwidth,height=3cm]{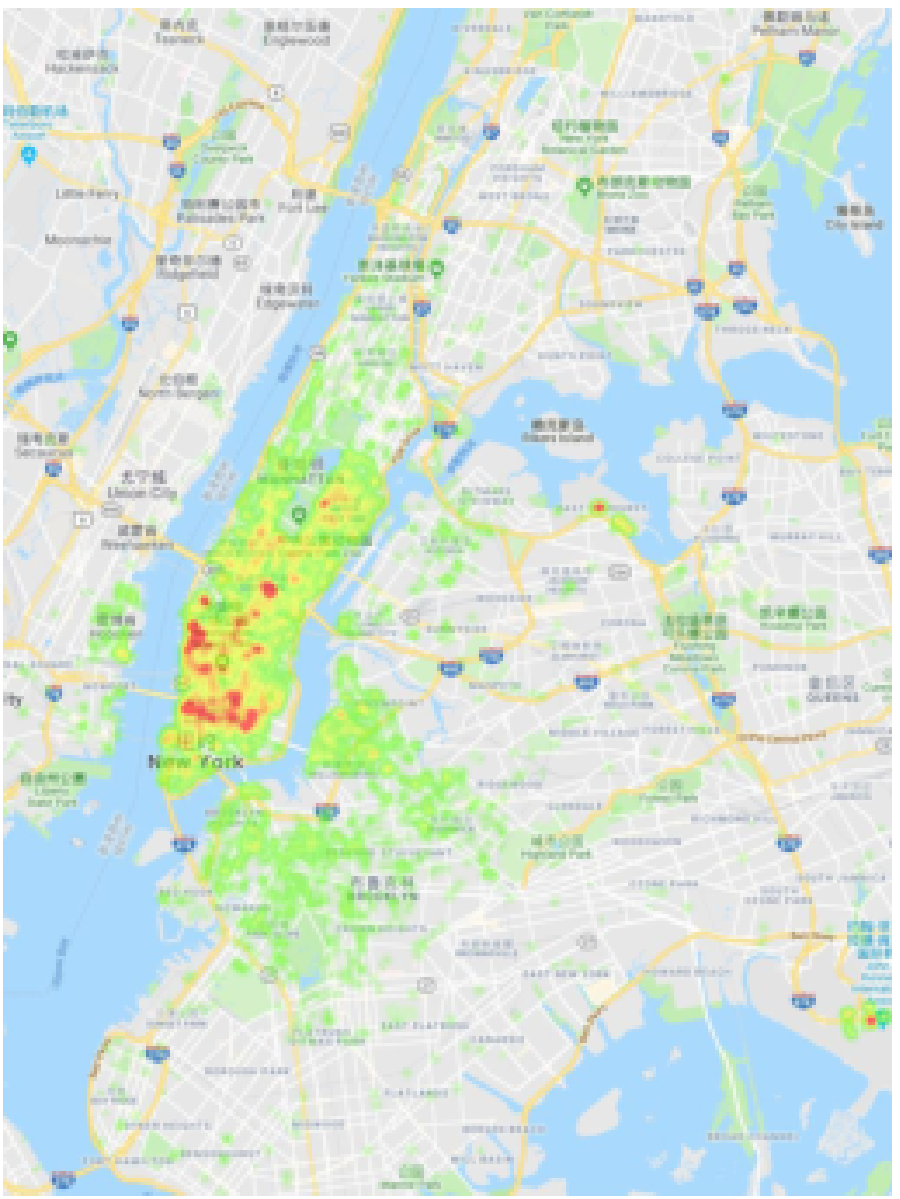}\label{fg:5-10}}
\subfigure[May 11(Sun)]{\includegraphics[width=0.3\textwidth,height=3cm]{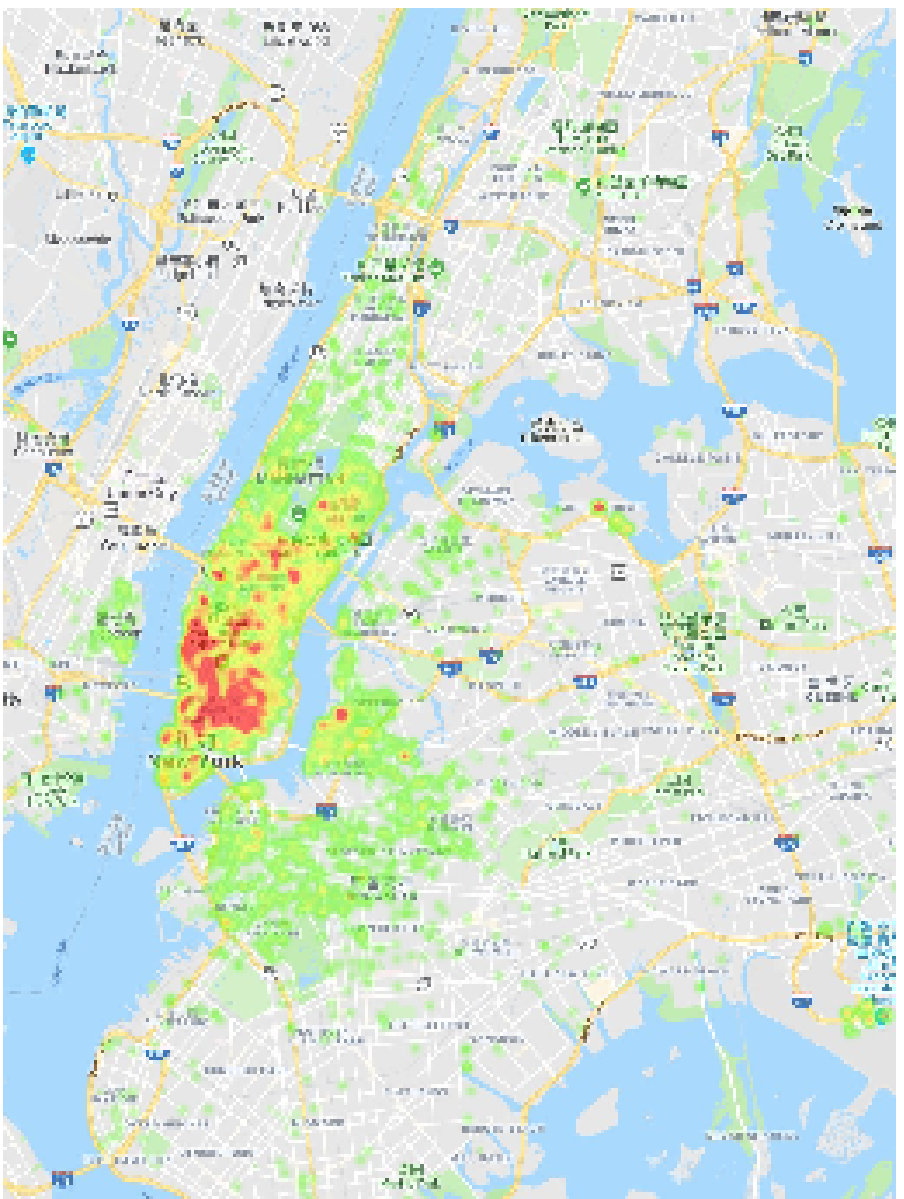}\label{fg:5-11}}
\subfigure[May 14(Wed)]{\includegraphics[width=0.3\textwidth,height=3cm]{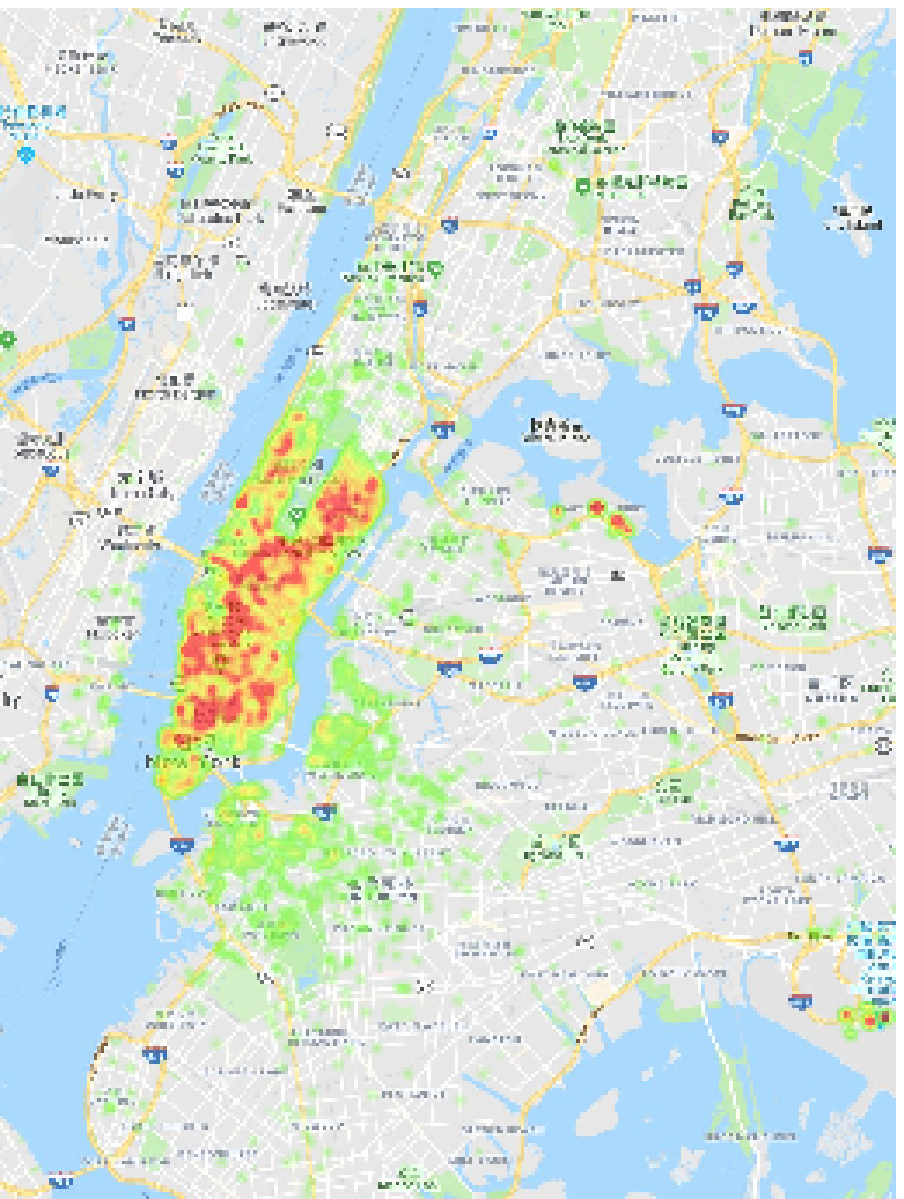}\label{fg:5-14}}
\caption{Similarity on the distribution of Uber pickups for different days in May, 2014}
\label{fg:distribution}
\end{figure}
We can observe that the spatiotemporal distributions of pickups are similarity each day. 
Although there is a large deviation in the quantity of pickups, the distributions of travel cost in optimal matching schemes for these days remain similar if pickups represent workers in matching with same batch of tasks. 
The main reason is that distribution plays crucial roles rather than quantity in budget-constraint matching problem, and superfluous workers would not produce significant impact on the final optimal matching scheme. 
Thus, it is feasible to utilize historical optimal matching scheme to guid the online assignment for the other days.
In this section, we propose another greed variant, called Greed-OT, which utilizes historical information to obtain proper threshold. Specifically, Greedy-OPT first runs the offline optimal algorithm based on workers' historical traces, and then from the matching result draws the maximum travel cost as target threshold. The detail description on Greedy-OT is listed in Algorithm \ref{al:ot}.
\begin{algorithm}[!h]
\caption{Greedy-OT Algorithm}
\label{al:ot}
\begin{algorithmic}[1]
\REQUIRE $B$,$\hat{M}^*$
\ENSURE matching scheme $M$
\STATE $c\gets0$, $\tau\gets$$\max\limits_{(w,t)\in\hat{M}^*}cost(w,t)$
\FORALL {new arrival worker $w$}
    \STATE $T^\prime\gets$\{$\forall t\vert b_w+cost(w,t)/v_w\le d_t\land cost(w,t)\le \tau\land c+cost(w,t)\le B$ \}
    \IF{$T^\prime\neq \emptyset$}
        \STATE $t=\min\limits_{t\in T^\prime}cost(w,t)$
        \STATE $M\gets M\cup$($w$,$t$)
        \STATE $c\gets c+cost(w,t)$
    \ENDIF
\ENDFOR
\RETURN $M$
\end{algorithmic}
\end{algorithm}
\begin{example}
Backing to our running example in Example \ref{ex:1}, we assume historical data has the same spatiotemporal distribution with this case. Greedy-OT first utilizes maximum travel cost of the optimal matching pairs as the threshold, here the maximum pair is $(w_4, t_5)$ with cost 5, and then prunes bad pairs whose cost exceeds 5. Thus, we can obtain the final matching solution:\\
$5$:$\{(w_2,t_2),(w_3,t_6),(w_5,t_1),(w_6,t_4)\}$, matching size:4, used budget:10.\\
In this example, Greedy-OT and OPT achieve the same score in matching size, but OPT is more frugal in budget expenditure for only 90 percent budget is spent while 100 percent does as for Greedy-OPT.
\end{example}
Next we deduce the competitive ratio of Greedy-OT. Suppose $n_i$ is the number of pairs with travel cost $c_i$ in the optimal set $\mathcal{O}$ which owns $N$ unique cost values. That is, the size of pairs $|\mathcal{O}|=\sum_{i=1}^N n_i$. The travel cost of any  pair from $w$ to $t$ in $\mathcal{O}$ can be denoted by $c_{(w,t)_{opt}}$, where $(w,t)_{opt}\in[1,N]$.
$W_\mathcal{O}$ is the worker set and $T_\mathcal{O}$ is the task set in $\mathcal{O}$. 
The maximum travel cost of pairs in $\mathcal{O}$ is denoted by $c_{max}^*$, and $c_{max}^*+\varepsilon$ is the threshold that Greedy-OT adopts to prune inferior matching pairs. 
\begin{lemma}\label{lm:overlap}
If there exists a pair$(w,t)$ chosen by Greedy-OT with its cost is little than $c_{max}^*$, then it must be $w\in W_\mathcal{O}$ or $t\in T_\mathcal{O}$.
\end{lemma}
\begin{proof}
Suppose there exits a pair $(w,t)\notin \mathcal{O}$ with its cost little than $c_{max}^*$, and its two ends, $w$ and $t$, follows $w\notin W_\mathcal{O}$ and $t\notin T_\mathcal{O}$, then by replacing the maximum cost pair with $(w,t)$, we would obtain a new optimal solution, which is contradict with the fact that $\mathcal{O}$ is the genuine optimal solution. The lemma is proved.
\end{proof}
\begin{lemma}\label{lm:cost}
If a pair $p$ chosen by Greedy-OT  is overlapped with a pair $p^*$ in $\mathcal{O}$ at least one point ($w$ or $t$), then $cost(p^*)\le cost(p)\le c_{max}^*+\varepsilon$.
\end{lemma}
\begin{proof}
Since Greedy-OT selects  $c_{max}^*+\varepsilon$ as the filter threshold, the cost of any chosen pair can not beyond $c_{max}^*+\varepsilon$. 
Assume the cost of $p$ is little than $p^*$, then $p$ should be added into $\mathcal{O}$ instead of $p^*$, which is contradict with the fact that $p^*$ is the genuine pair in $\mathcal{O}$. Thus, the assumption that $cost(p)<cost(p^*)$ is false and alternatively $cost(p)\ge cost(p^*)$. The lemma is proved.
\end{proof}
\begin{theorem}\label{th2}
The competitive ratio of Greedy-OT is not less than $\frac{\sum_{i=1}^N c_i n_i}{(c_{max}^*+\varepsilon)\cdot\sum_{i=1}^N n_i}$.
\end{theorem}
\begin{proof}
For the case $\varepsilon\ge 0$, according to Lemma \ref{lm:overlap} and \ref{lm:cost}, the cost of pair $(w, t)$ chosen by Greedy-OT affiliates with either of the following cases: \\
1)$cost(w,t)\in[c_{i},c_{max}^*+\varepsilon]$, where $w\in W_\mathcal{O}$, $cost(w,)_{opt}=c_i$ or $t\in T_\mathcal{O}$, $cost(,t)_{opt}=c_i$;\\
2)$cost(w,t)\in[c_{max}^*,c_{max}^*+\varepsilon]$, where $w\notin W_\mathcal{O}$ and $t\notin T_\mathcal{O}$.\\
For both of cases, any pair $(w,t)$ that Greedy-OT chooses can be mapped to an optimal pair in $\mathcal{O}$, which inflates the cost from optimal value $c$ to $cost(w,t)\in[c,c_{max}^*+\varepsilon]$ or $[c_{max}^*,c_{max}^*+\varepsilon]$ and consequently decreases the matched number of Greedy-OT. 
Specifically, for the case 1, pair $(w,t)$ is bound to be mapped to a optimal pair whose worker entry is $w$ or task entry is $t$, and for the case 2, pair $(w,t)$ is mapped to a random optimal pair which will never be selected by the end of Greedy-OT. 
Suppose $\xi_{ij}$ is the travel cost of chosen pair which is mapped to the $j$-th pair among $n_i$ optimal pairs with the same cost $c_i$, then the lower bound of $\xi_{ij}$ is either $c_i$ (case 1) or $c_{max}^*$ (case 2) while the upper bound are both $c_{max}^*+\varepsilon$. 
The worst expectation of Greedy-OT is
\begin{equation}
\begin{split}
E(Greedy-OT)&=E(\frac{B}{cost})=B\cdot E(\frac{1}{cost})\\
&=\frac{B}{\sum_{i=1}^N n_i}(\sum_{i=1}^N\sum_{j=1}^{n_i}\frac{1}{\xi_{ij}})\\
&\ge \frac{B}{\sum_{i=1}^N n_i}(\sum_{i=1}^N {n_i}\cdot \frac{1}{c_{max}^*+\varepsilon})\\
&=\frac{B}{c_{max}^*+\varepsilon}
\end{split}
\end{equation}
Since $B\ge\sum_{i=1}^N c_i n_i$, we have
\begin{equation}
\begin{split}
CR(Greedy-OT)&=\frac{E(Greedy-OT)}{|\mathcal{O}|}\\
&\ge \frac{\sum_{i=1}^N c_i n_i}{(c_{max}^*+\varepsilon)\cdot \sum_{i=1}^N n_i}.
\end{split}
\end{equation}

For the case $\varepsilon< 0$, any chosen pair $(w, t)$ by Greedy-OT is bound to overlap with $\mathcal{O}$ according to Lemma \ref{lm:overlap}. The travel cost of $(w, t)$ only affiliates one case, that is\\ 
1)$cost(w,t)\in[c_{i},c_{max}^*+\varepsilon]$, where $w\in W_\mathcal{O}$, $cost(w,)_{opt}=c_i$ or $t\in T_\mathcal{O}$, $cost(,t)_{opt}=c_i$.\\\ 
Suppose $c_K\le c_{max}^*+\varepsilon<c_{K+1}\le c_N$,  then we have
\begin{equation}
\begin{split}
E(Greedy-OT)&=\frac{B}{\sum_{i=1}^K n_i}(\sum_{i=1}^K\sum_{j=1}^{n_i}\frac{1}{\xi_{ij}})\\
&\approx \frac{B}{\sum_{i=1}^K n_i}(\sum_{i=1}^K {n_i}\cdot E(\frac{1}{\xi_{ij}})),
\end{split}
\end{equation}
where $\xi_{ij}\in[c_i,c_{max}^*+\varepsilon]$. Here
\begin{equation}
\begin{split}
E(\frac{1}{\xi_{ij}})&=\int_{c_i}^{c_{max}^*+\varepsilon} \frac{1}{\xi_{ij}}\cdot \frac{1}{c_{max}^*+\varepsilon-c_i}d{\xi_{i}}\\
&=\frac{\ln{(c_{max}^*+\varepsilon)}-\ln{c_i}}{c_{max}^*+\varepsilon-c_i}. 
\end{split}
\end{equation}
Since $B\ge\sum_{i=1}^N c_i n_i$, we have
\begin{equation}
\begin{split}
CR(Greedy-OT)&=\frac{E(Greedy-OT)}{|\mathcal{O}|}\\
&\ge \frac{(\sum_{i=1}^N c_i n_i)\cdot (\sum_{i=1}^K {n_i}\frac{\ln{(c_{max}^*+\varepsilon)}-\ln{c_i}}{c_{max}^*+\varepsilon-c_i})}{\sum_{i=1}^N n_i\cdot \sum_{i=1}^K n_i}.
\end{split}
\end{equation}
On the ground that $f(x)=\frac{\ln{(c_{max}^*+\varepsilon)}-\ln{x}}{c_{max}^*+\varepsilon-x}$ is a decreasing function on $x$, and $\lim_{x\to c_{max}^*+\varepsilon}f(x)=\frac{1}{c_{max}^*+\varepsilon}$, we have
\begin{equation}
\begin{split}
CR(Greedy-OT)&\ge \frac{\sum_{i=1}^N c_i n_i}{(c_{max}^*+\varepsilon)\cdot\sum_{i=1}^N n_i}.
\end{split}
\end{equation}
The theorem follows.
\end{proof}
\section{Experiments}
\subsection{Experiment Setup}
$\textbf{Synthetic Dataset}$. We generate 10000 workers and 10000 tasks on a 500$\times$500 2D square where all positions are generated randomly. 
The value of $c_{max}$ is set to be 1000 which is the maximum Manhattan distance in the specific area.
The appearance of workers follows adversary model and random model separately, and in either of the models, the arrival time of workers and the release time of tasks scatter randomly between 0 and 99. 
The measurement on a worker's travel cost equals to the Manhattan distance between his present and target location. 
For simplicity, we set all workers' velocity to 1, so that time cost is equivalent to distance cost.
Platform expends budget to reward any assigned worker according his travel cost. 
The statistics and configuration of synthetic data are illustrated in Table \ref{tb:sd}, where the default settings are marked in bold.

$\textbf{Real Dataset}$. We choose Uber Trip Data \cite{Uber} as our real dataset, which contains data on over 4.5 million Uber pickups in New York City from April to September 2014, and 14.3 million more Uber pickups from January to June 2015. 
Based on the raw data of the second week in May, 2014, we select the rectangle area from the position(Lat:40.5998, Lon:-74.0701) to the position(Lat:40.8998, Lon:-73.7701) as investigative area whose maximum distance $c_{max}$ is 41.7027km, which equals the Euclidean distance of its diagonal.
Any pickup has appeared from 0 o'clock to 12 o'clock in the area plays as a worker and the velocity of workers are set to the same value 40km/h.
Due to the lack of task information in UTD dataset, we randomly generate 6000 tasks whose location is limited in the area, release time randomly scatters from 0:00 to 12:00 and 180 minutes survival time before deadline. 
We extract the optimal threshold from the pickups data in May 7(Wed),2014, and then utilize it to guide the online matching for other days, including May 6(the previous day), May 8(the next day), May 10(weekend), and additional May 14(Wed of the next week). 
In our case, the requester supplies budget in amount of 300 to reward workers, which means the total travel cost (distance cost) of chosen workers can not beyond 300km.

We compare Greedy-RT, Greedy-OT with the baseline Greedy and the offline optimal OPT in terms of total matching size, running time and memory cost, and study the effect of varying parameters.
All the algorithms are implemented in C++, and run in a machine with Intel(R) Core(TM) i5-2400 CPU and 4GB main memory. 
\begin{table}
\caption{Synthetic Dataset}\label{tb:sd}
\centering
\begin{tabular}{|c|c|c|} \hline
Factor&Setting&Notes\\\hline
$|W|$&2000,4000,$\bm{6000}$,8000,10000&No. of workers\\\hline
$|T|$&2000,4000,$\bm{6000}$,8000,10000&No. of tasks\\\hline
$B$&1000,2000,$\bm{3000}$,4000,5000&total budget\\\hline
$d_t$&20,40,$\bm{60}$,80,100&Deadline of task\\\hline
\end{tabular}
\end{table}
\subsection{Adversary Model}
\textbf{Effect of $|W|$}. 
The first column of Fig. \ref{fg:ad} shows the results when $|W|$ varies from 2000 to 10000. 
We can observe that for all the algorithms except the simple greedy, the quantity of successful matched pairs increases as $|W|$ increases, which is natural as there are more competent pairs that can be matched when more workers are available. 
Only the curve of simple greedy declines when $|W|$ increases since pairs with larger cost keep growth in quantity and have prior in matching in adversary model. 
Among all online algorithms, Greedy-OT performs the best, followed by Greedy-RT and simple greedy. 
The reasons are as follows: 
1) Greedy-OT can achieve proper threshold from a set of offline optimal pairs. 
2) part of thresholds have poor performances in matching and thus the expected performance of Greedy-RT is affected correspondingly. 
3) simple greedy is seriously trapped in local optimal solutions in adversary model. 
Note that Greedy-RT maybe perform worse than simple greedy when less workers are available, on the ground that large number of proper pairs are filtered out by excessively small thresholds.
As for the running time and travel cost,  OPT expends much than online algorithms for the reason that min-cost max-flow graph algorithm needs much time and storage consumption on repetitively running Dijkstra algorithm and keeping running states. 
Other algorithms are subordinate to greedy algorithm family so that they are similar in terms of time and storage cost. Moreover, they are usually unrelated to the number of workers for the reason that algorithms will terminal in advance when budget is drained rather than when all workers have appeared. 

\textbf{Effect of $|T|$}. 
The second column of Fig. \ref{fg:ad} presents the results when we vary $|T|$ from 2000 to 10000. 
In terms of the quantity of matched pairs, we observe that the curves of all algorithm ascend when $|T|$ increase since there are more competent pairs to be matched. 
Greedy-OT still performs much better than Greedy-RT and simple greedy. 
Due to the number of workers, here is 6000, is abundant enough to prevent the decline in quantity that excessively small thresholds bring about. Thus, Greedy-RT could keep higher in quantity than simple greedy when $|T|$ increases.
The results on time and memory cost have similar patterns to those of varying $|W|$, and we omit the detailed discussion. 

\textbf{Effect of budget}. 
The third column of Fig. \ref{fg:ad} shows the results when we vary $B$ from 1000 to 5000. 
We can observe that the matching quantity rises when budget increases. The maximum value of optimal pairs grows when budget increases, which leads the threshold of Greedy-OT increases as well. 
Since the feature that workers with large cost appear early in adversary model, the ratio between Greedy-OT and OPT in matching size decreases when budget increases.
The running time of OPT rises with the increase of budget. The reason is that additional budget would cause optimal pairs grows in quantity. For each of pair, OPT will run Dijkstra's algorithm in a huge graph which needs expensive time cost.
As for online algorithms, their running time is still tiny even can be ignored when compared with OPT.
The memory cost of OPT has no related to the budget since they run on the identical graph structure with same vertices and edges. 

\textbf{Effect of deadline}.
The fourth column of Fig. \ref{fg:ad} depicts the results of varying tasks' deadline from 20 to 100. 
All algorithms except the simple greedy increase in matching size with the increase of deadline.
This is because that long deadline represents more candidate tasks for a new arrival worker to match, which leads to less travel cost and correspondingly a larger matching size.
The simple greedy keeps invariable in match size for the reason that the travel cost of all matched pairs is little than 20, which means all workers can reach the nearest task before its deadline.
The results on time and memory cost have similar patterns to those of varying budget, and we omit the detailed discussion.
\begin{figure} 
\centering
\subfigure[Matching size of varying $|W|$]{\includegraphics[width=0.23\textwidth]{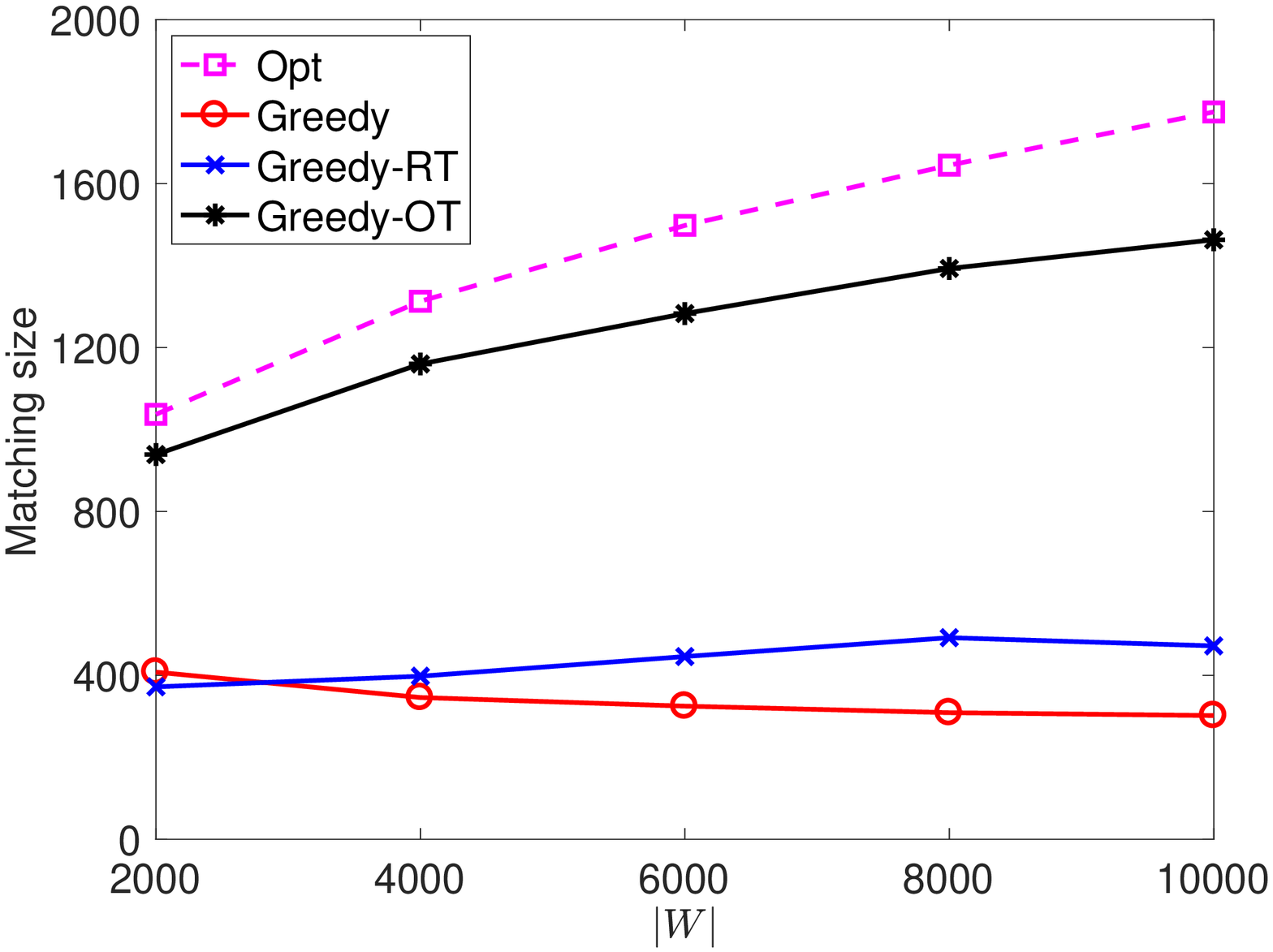}\label{fg:n_worker}}
\subfigure[Matching size of varying $|T|$]{\includegraphics[width=0.23\textwidth]{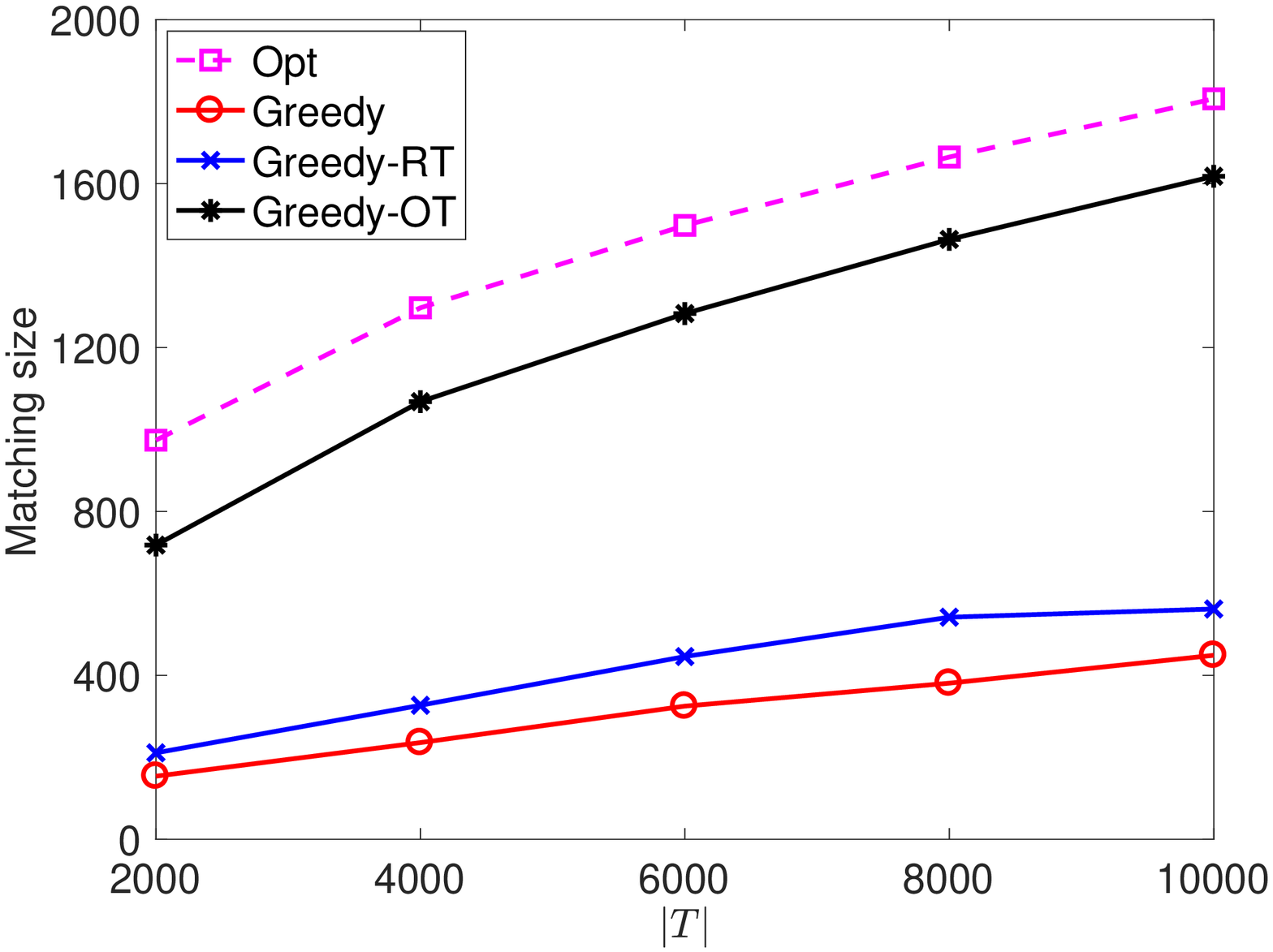}\label{fg:n_task}}
\subfigure[Matching size of varying $B$]{\includegraphics[width=0.23\textwidth]{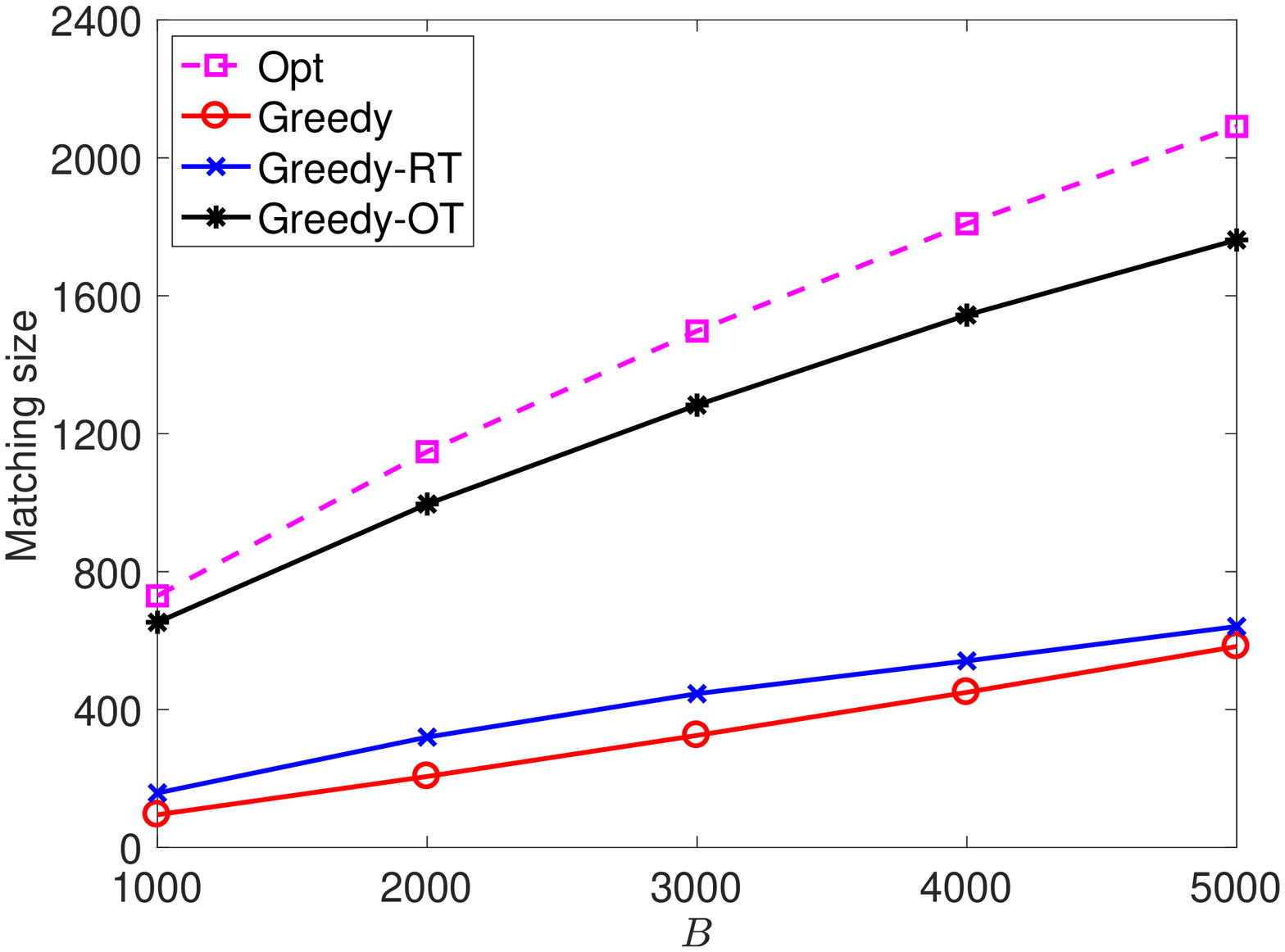}\label{fg:n_budget}}
\subfigure[Matching size of varying $d_t$]{\includegraphics[width=0.23\textwidth]{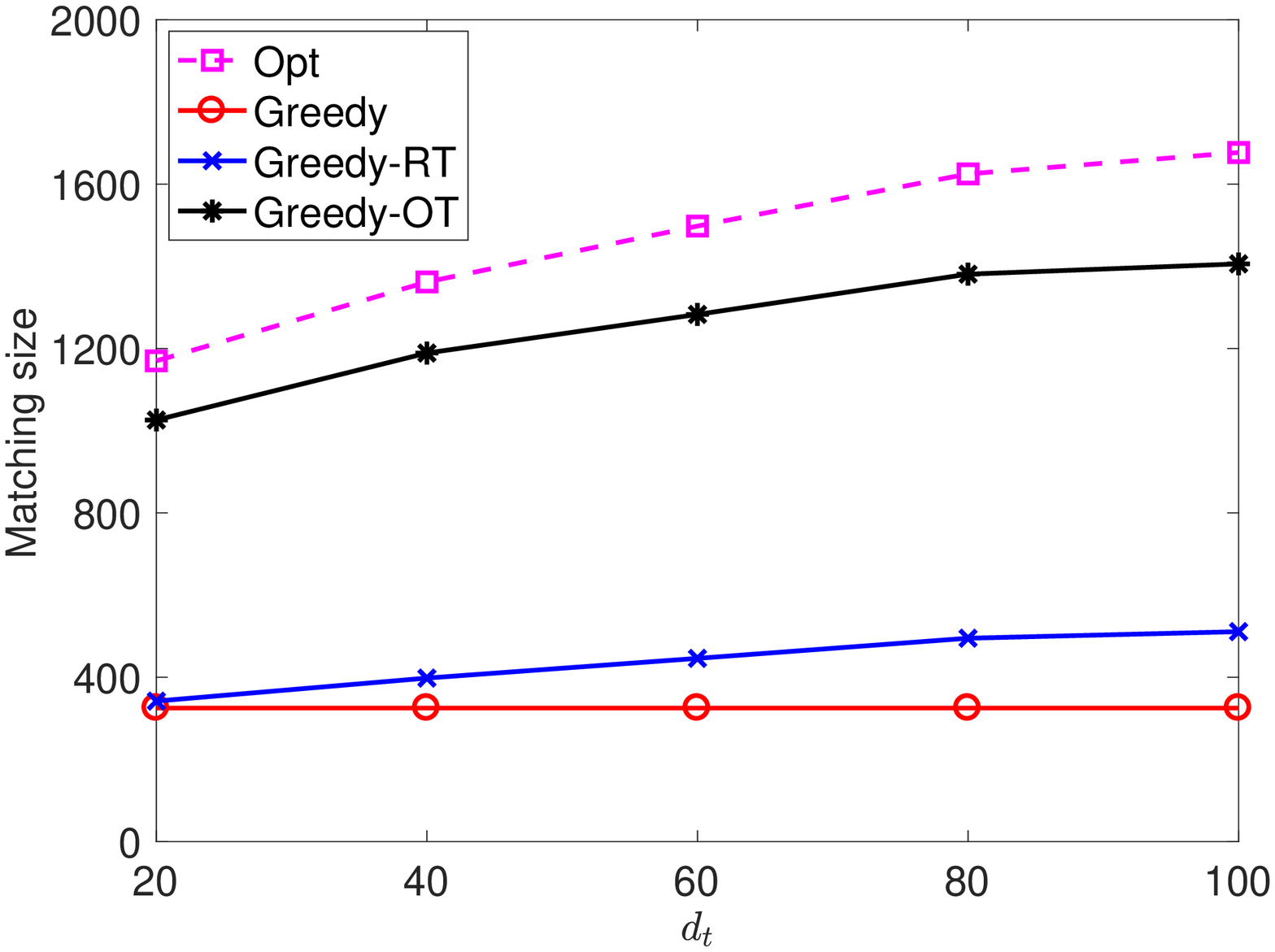}\label{fg:n_deadline}}
\vfill
\subfigure[Time of varying $|W|$]{\includegraphics[width=0.23\textwidth]{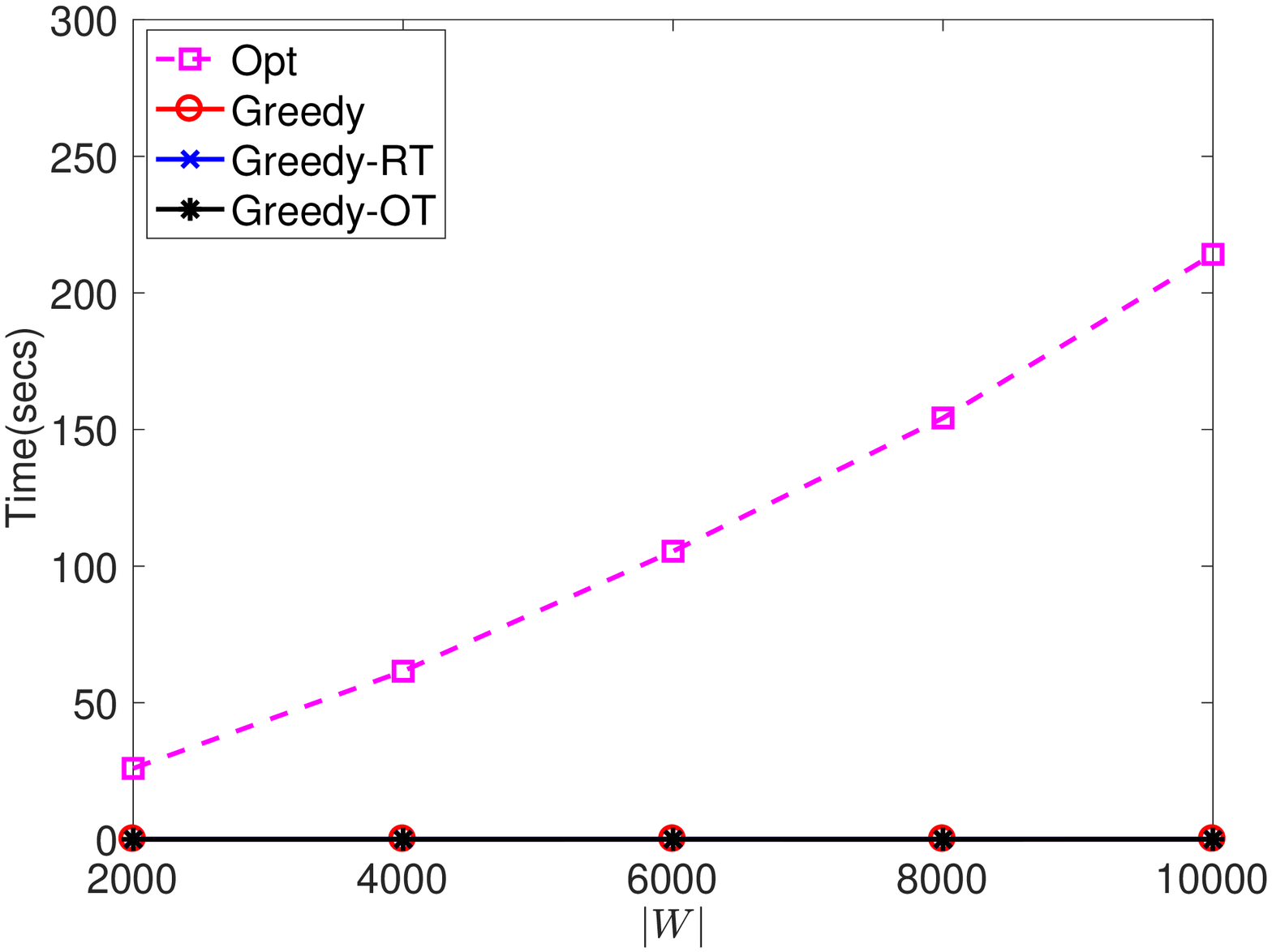}\label{fg:t_worker}}
\subfigure[Time of varying $|T|$]{\includegraphics[width=0.23\textwidth]{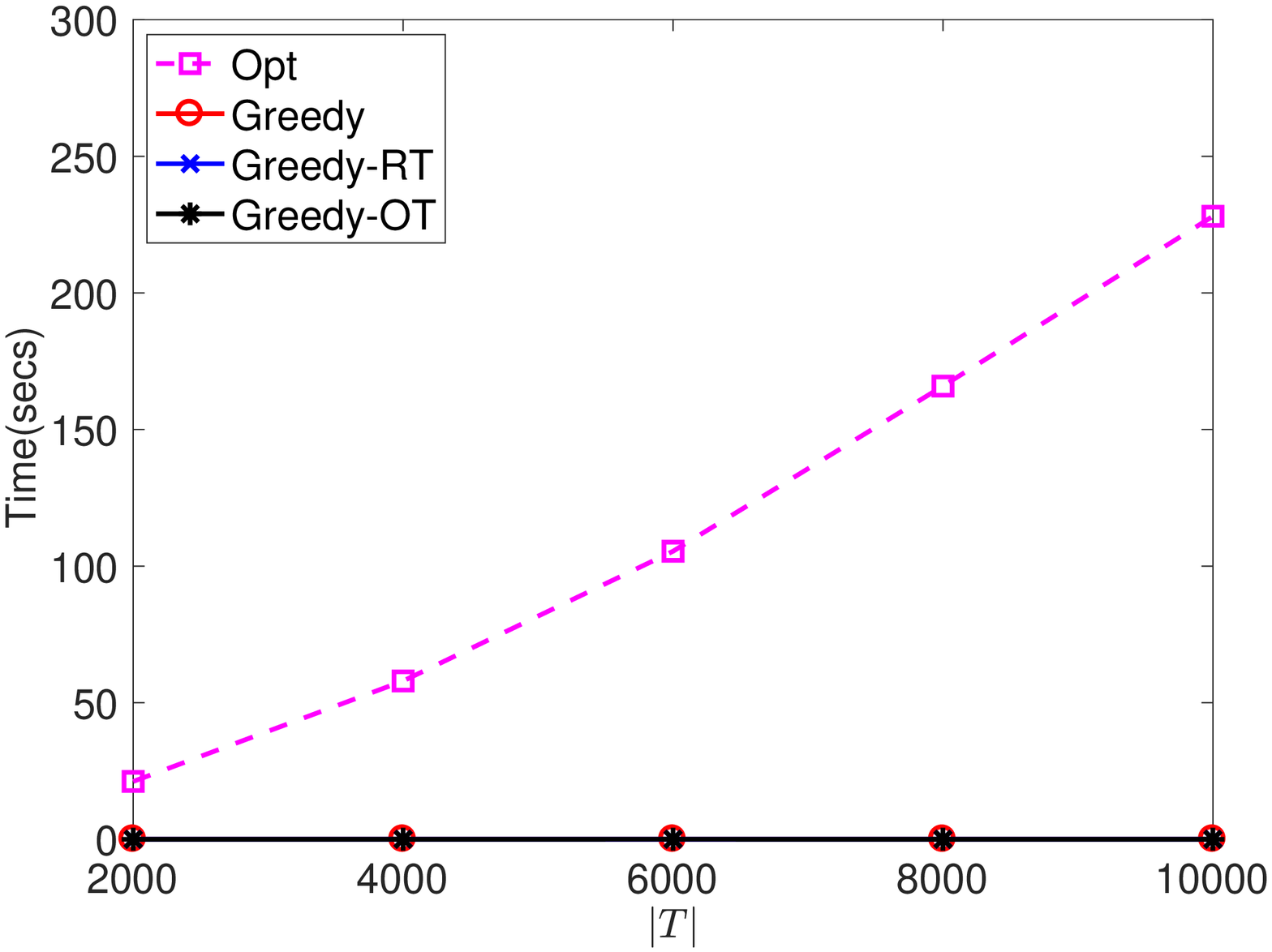}\label{fg:t_task}}
\subfigure[Time of varying $B$]{\includegraphics[width=0.23\textwidth]{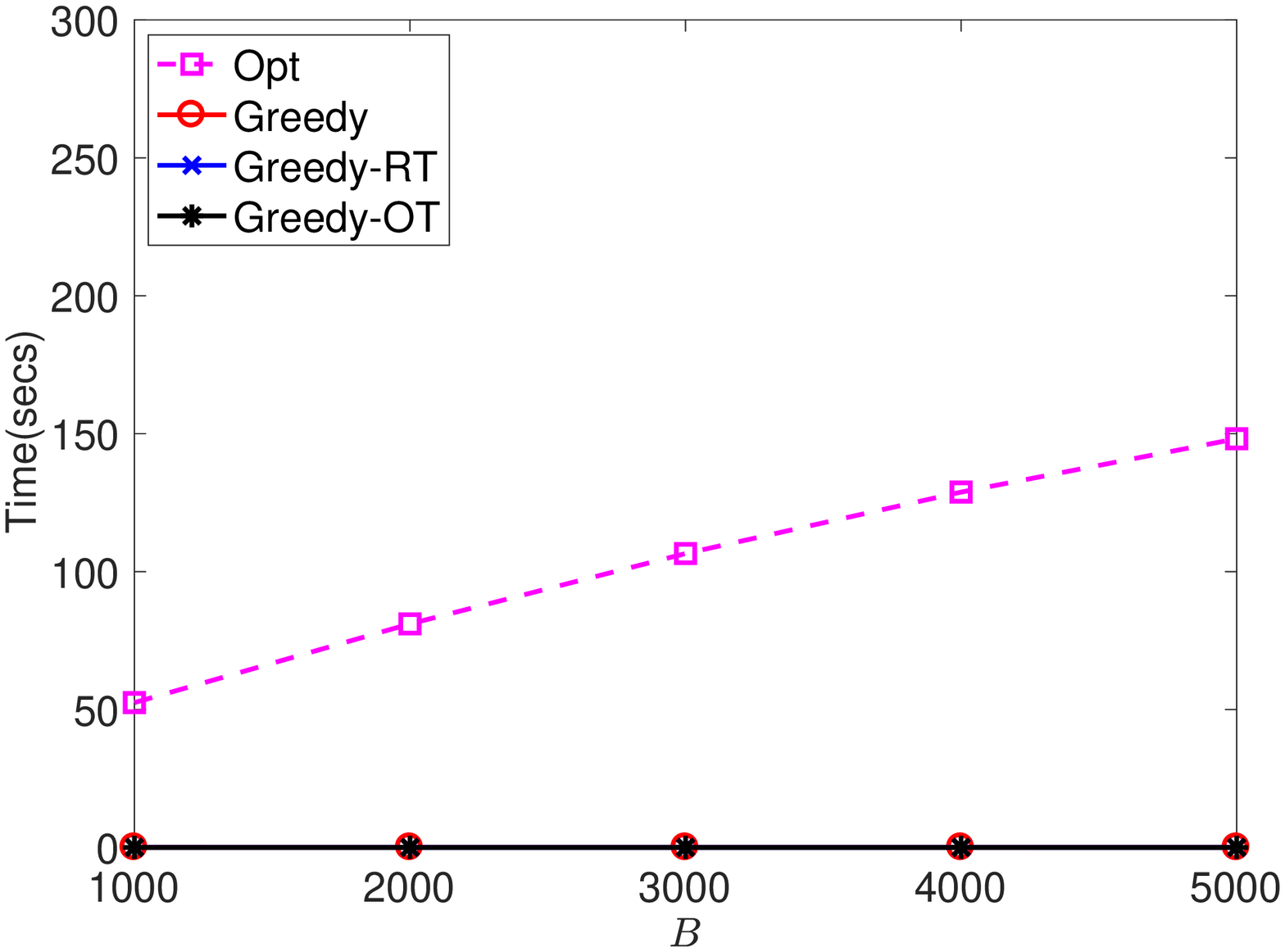}\label{fg:t_budget}}
\subfigure[Time of varying $d_t$]{\includegraphics[width=0.23\textwidth]{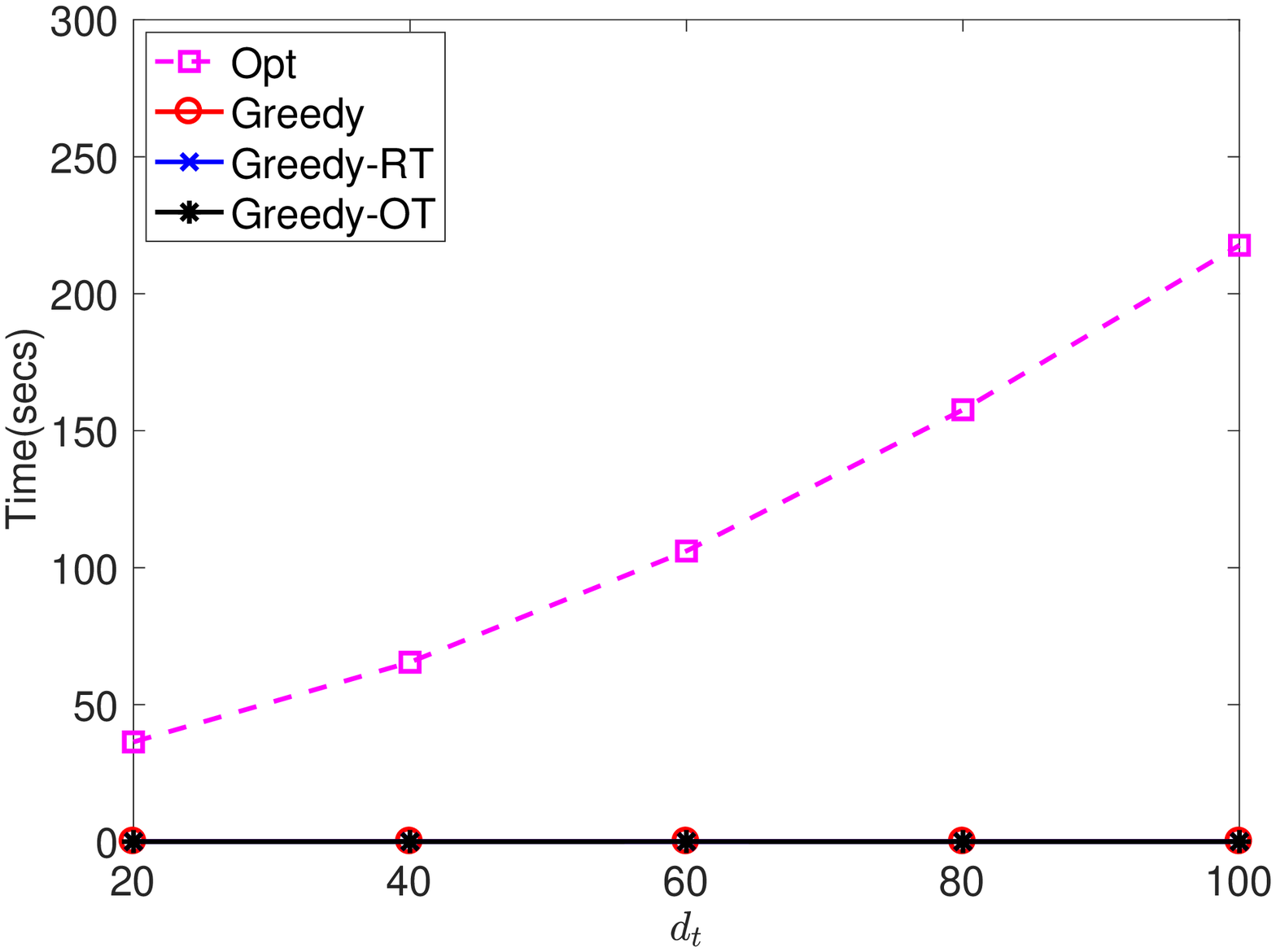}\label{fg:t_deadline}}
\vfill
\subfigure[Memory of varying $|W|$]{\includegraphics[width=0.23\textwidth]{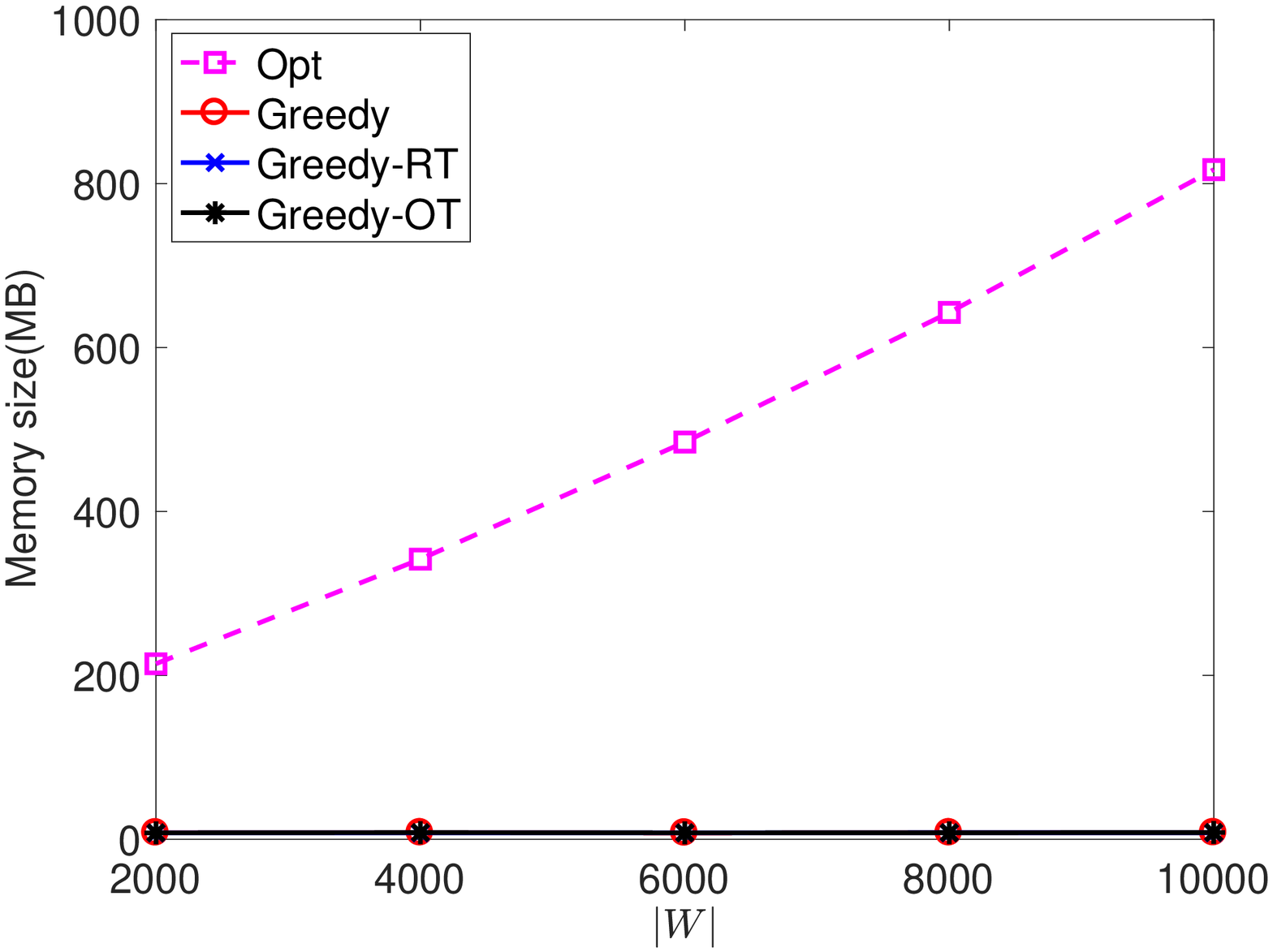}\label{fg:m_worker}}
\subfigure[Memory of varying $|T|$]{\includegraphics[width=0.23\textwidth]{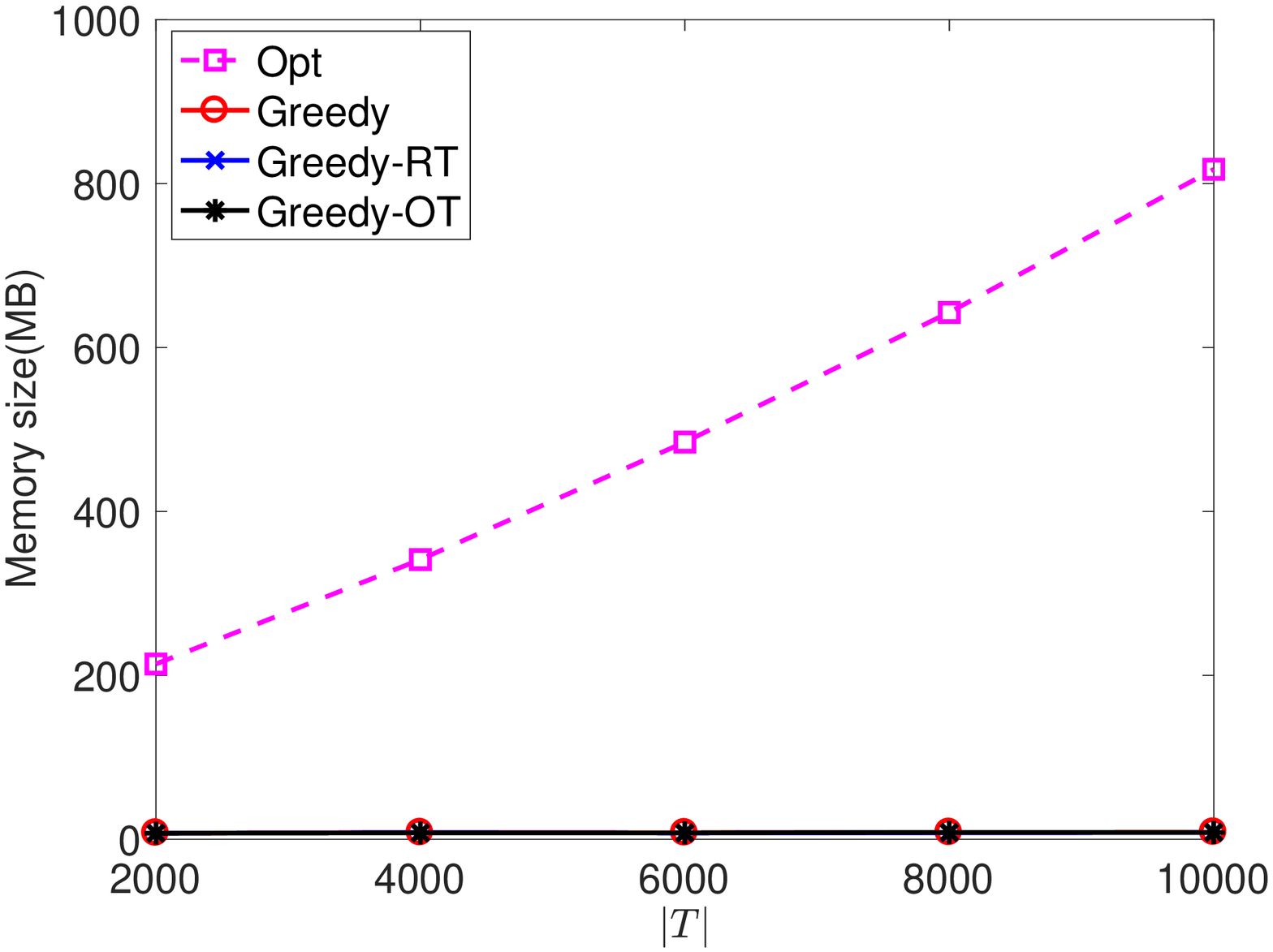}\label{fg:m_task}}
\subfigure[Memory of varying $B$]{\includegraphics[width=0.23\textwidth]{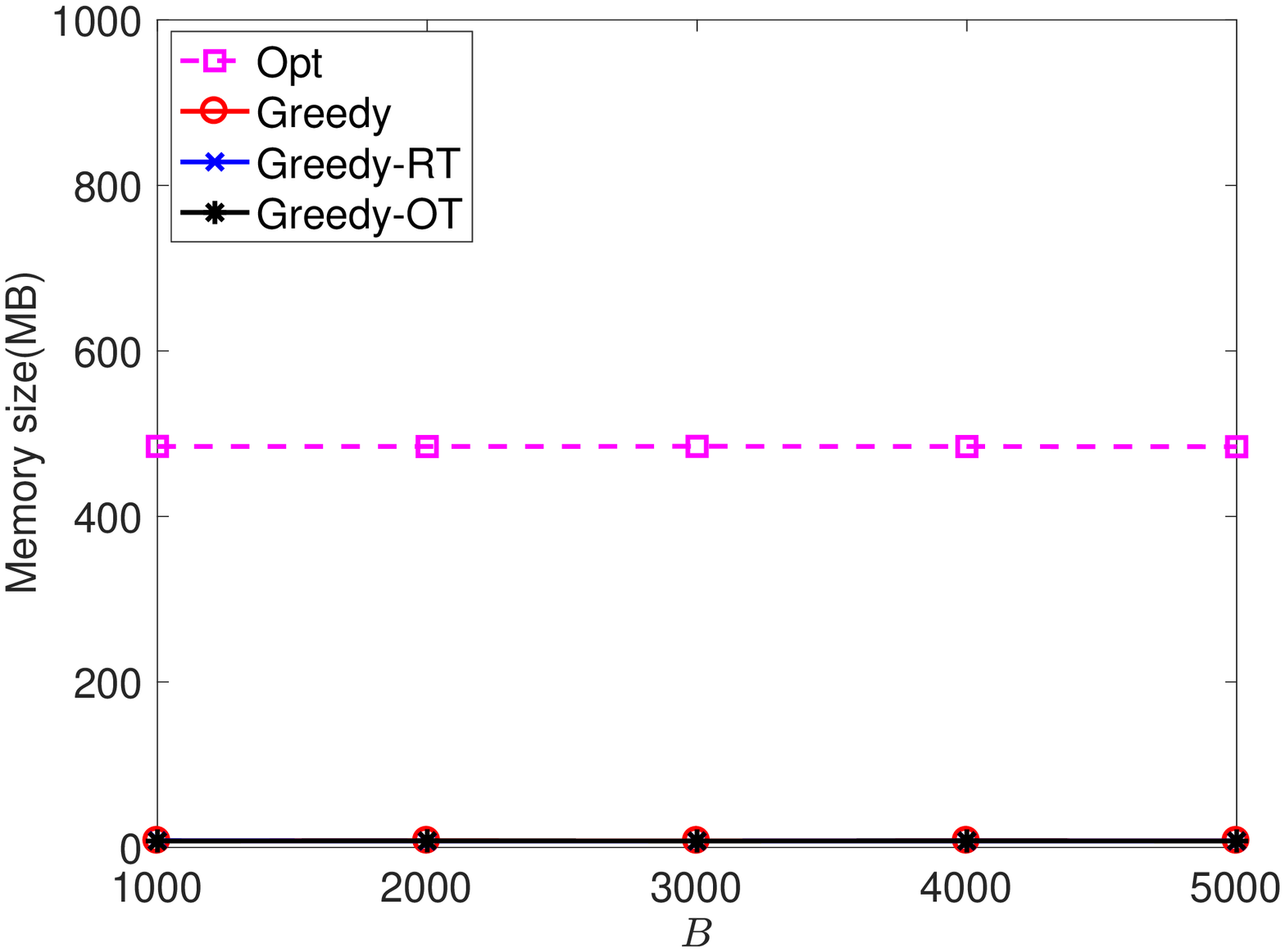}\label{fg:m_budget}}
\subfigure[Memory of varying $d_t$]{\includegraphics[width=0.23\textwidth]{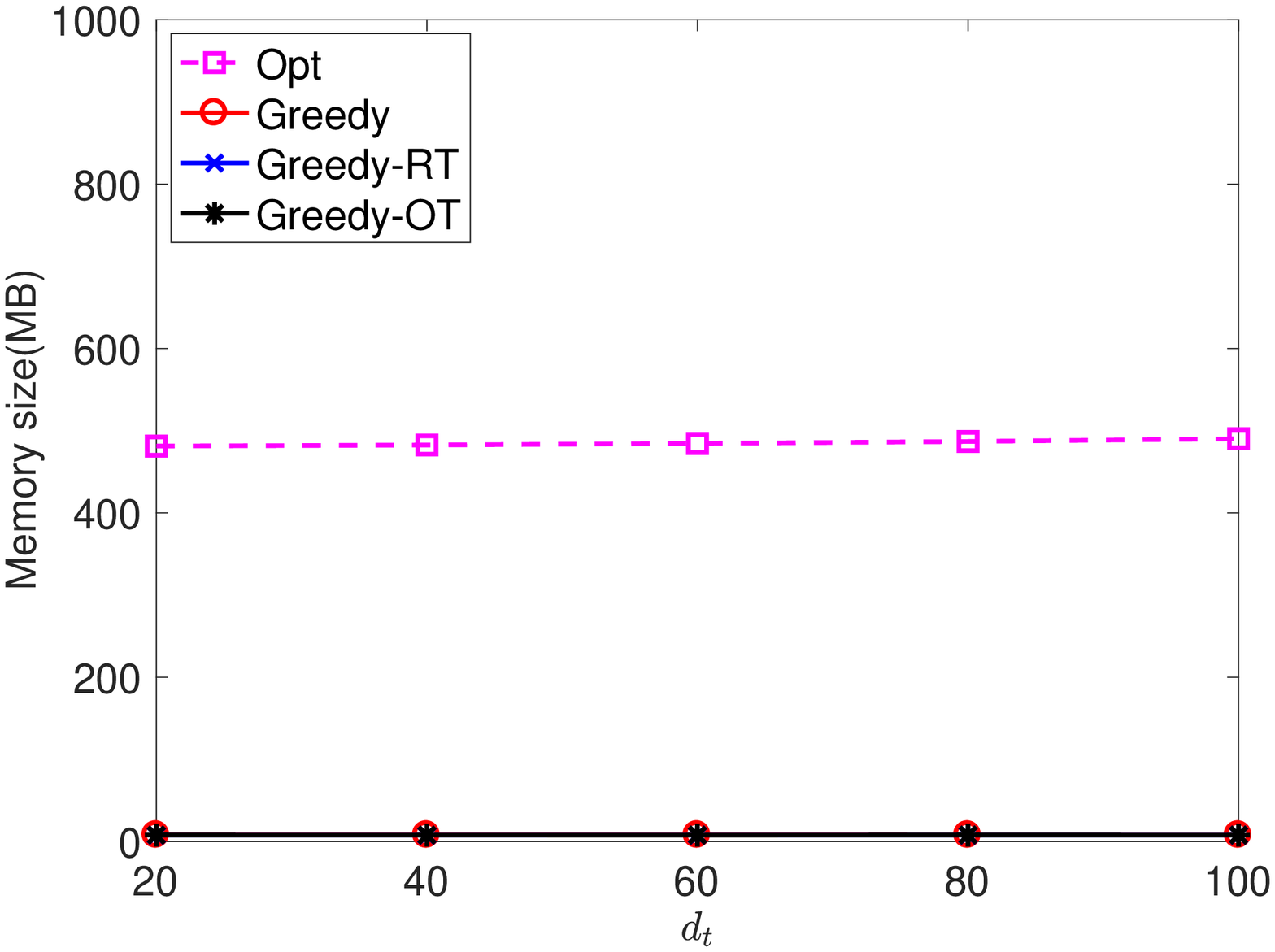}\label{fg:m_deadline}}
\caption{Results on varying each specific parameter in adversary model}
\label{fg:ad}
\end{figure}
\subsection{Random Model}
In this section, we will analysis the performance of all algorithms run in random model. 
As Fig. \ref{fg:rd} depicts, all algorithms perform highly similar to adversary model in terms of matching size, so we merely analysis the discrepancies between Fig. \ref{fg:ad} and \ref{fg:rd}.
First, we can observe that the curve of Greedy-OT in random model is more asymptotic to the optimal curve than in adversary model, which means the Greedy-OT performs better in random model. 
Since adversary model has the feature that high cost workers have prior in appearance, workers whom Greedy-OT chooses have larger travel cost which consumes budget highly and consequently decreases matching size.
Second, the gap of matching size between Greedy-RT and simple greedy is more narrow in random model. 
The main reason is that simple greedy improves more significant than Greedy-RT in random model in terms of matching size.
Third, compared with Greedy-RT, simple greedy may perform better in certain conditions such as less available workers or sufficient budget. 
Specifically, as Fig. \ref{fg:n_worker} shown, Greedy-RT works well when $|W|$ equals 2000 and 4000. The reasons are as follows: 1) small thresholds Greedy-RT selects decrease its matching size; 2) simple greedy performs better in random model than adversary model. We also find that simple greedy outperforms than Greedy-RT when budget beyonds 4000 from Fig. \ref{fg:n_budget}. The reasons are as follows: 1) small thresholds still filter out many proper pairs even if the budget is sufficient. 2) those well-performance thresholds promote matching size slowly, or even reach the limits of matching size.
The results on time and memory cost in random model are similar to adversary model, and we omit the detailed discussion.
\begin{figure}
\centering
\subfigure[Matching size of varying $|W|$]{\includegraphics[width=0.23\textwidth]{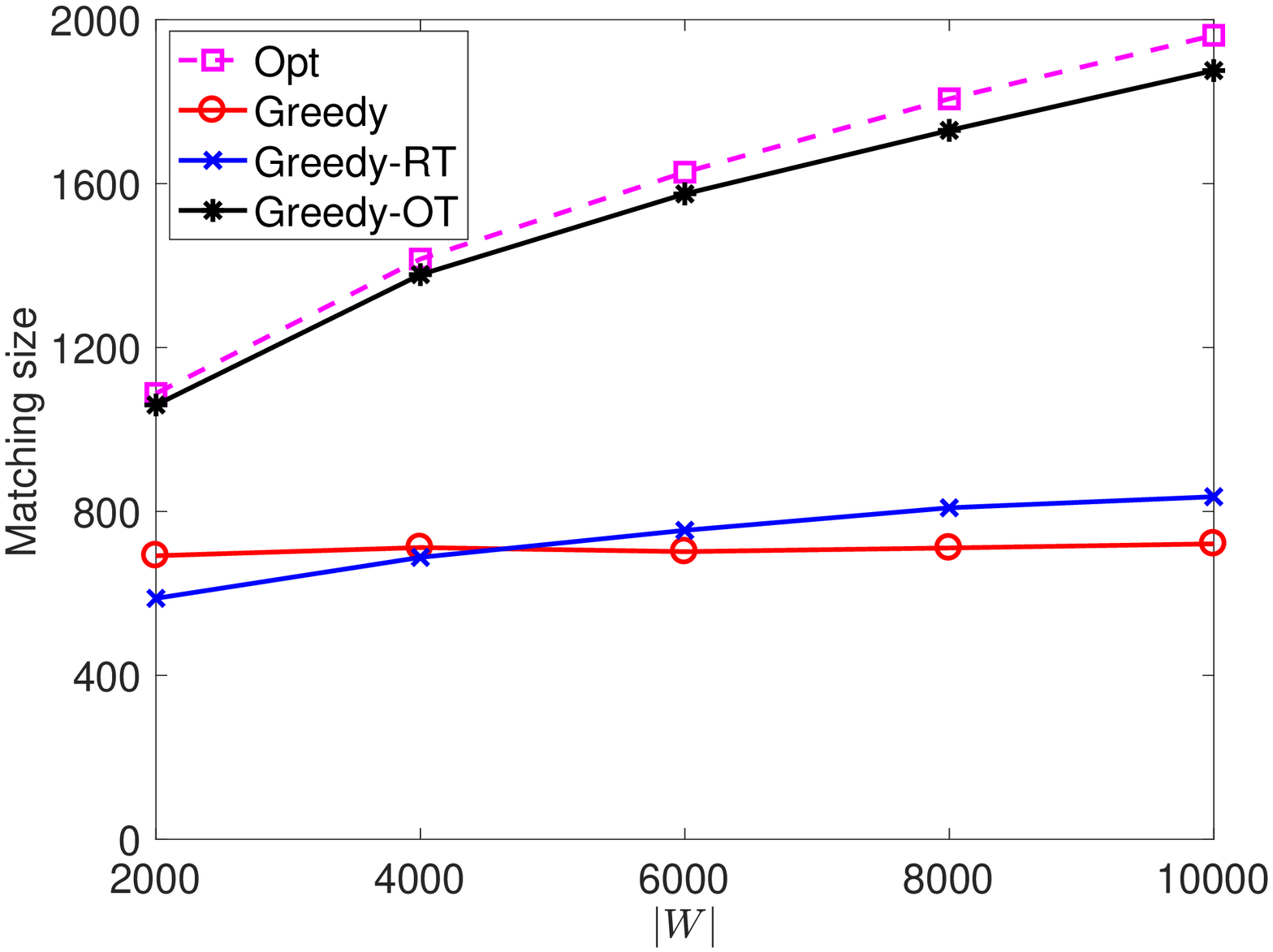}\label{fg:n_worker}}
\subfigure[Matching size of varying $|T|$]{\includegraphics[width=0.23\textwidth]{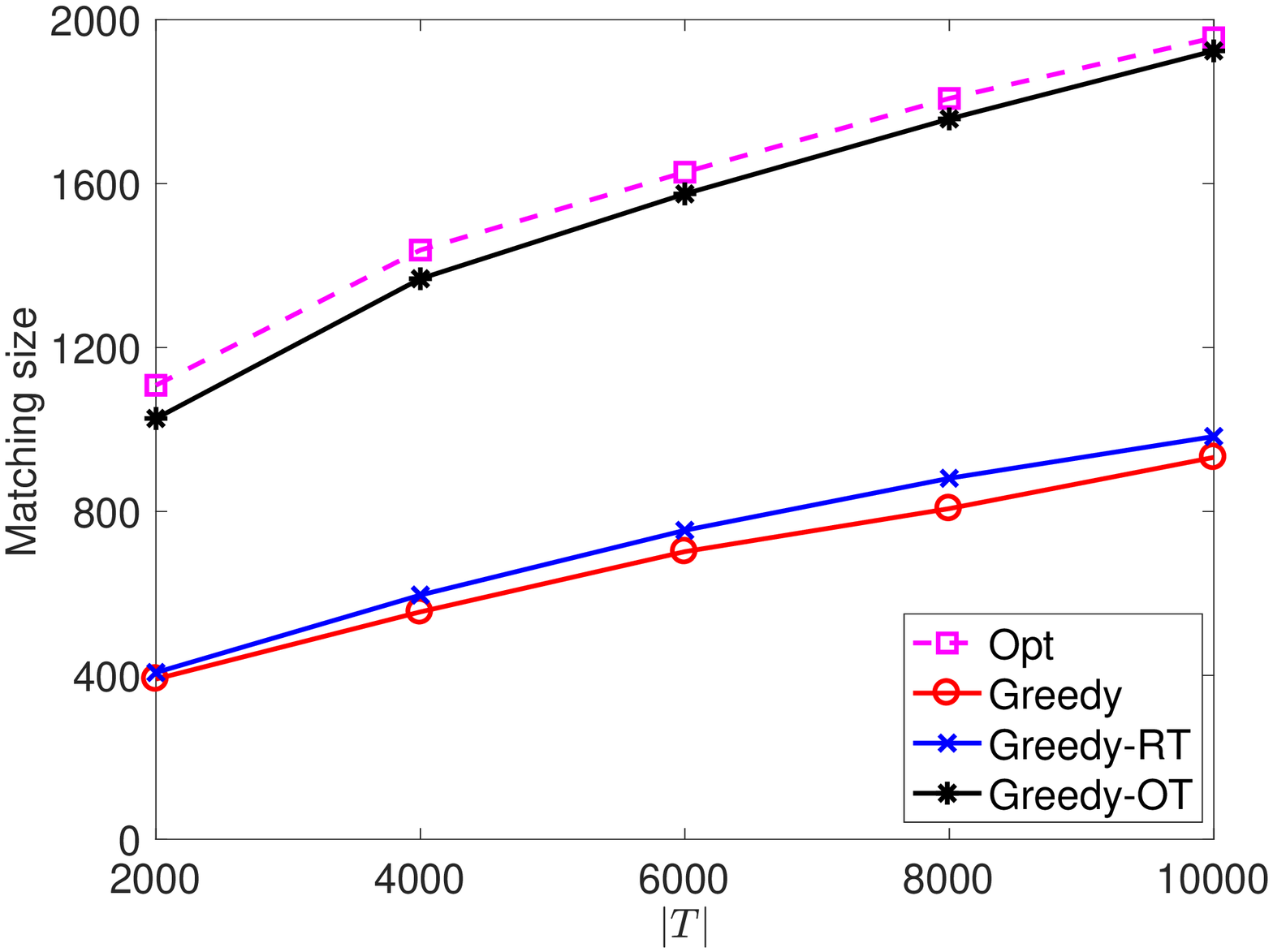}\label{fg:n_task}}
\subfigure[Matching size of varying $B$]{\includegraphics[width=0.23\textwidth]{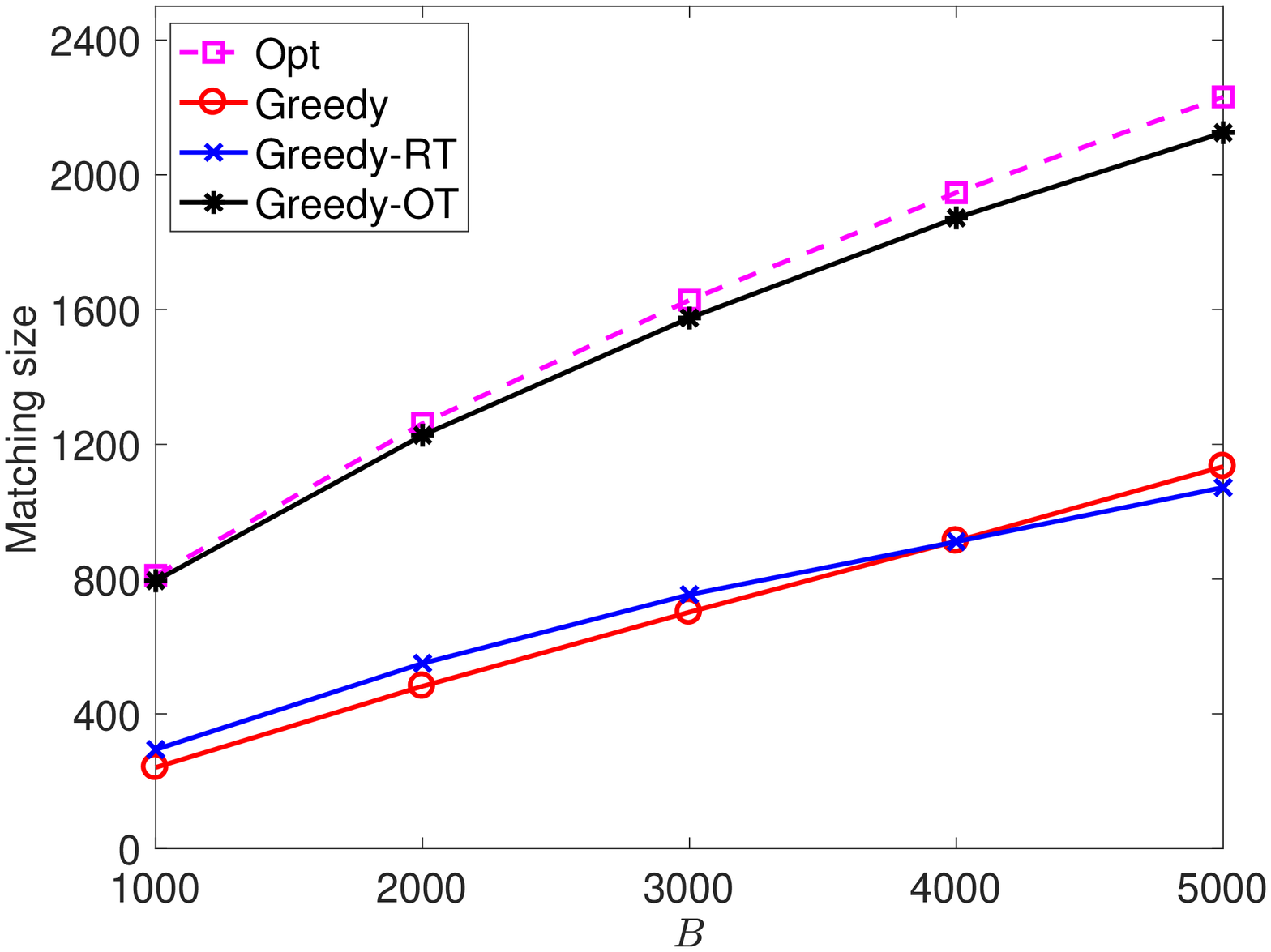}\label{fg:n_budget}}
\subfigure[Matching size of varying $d_t$]{\includegraphics[width=0.23\textwidth]{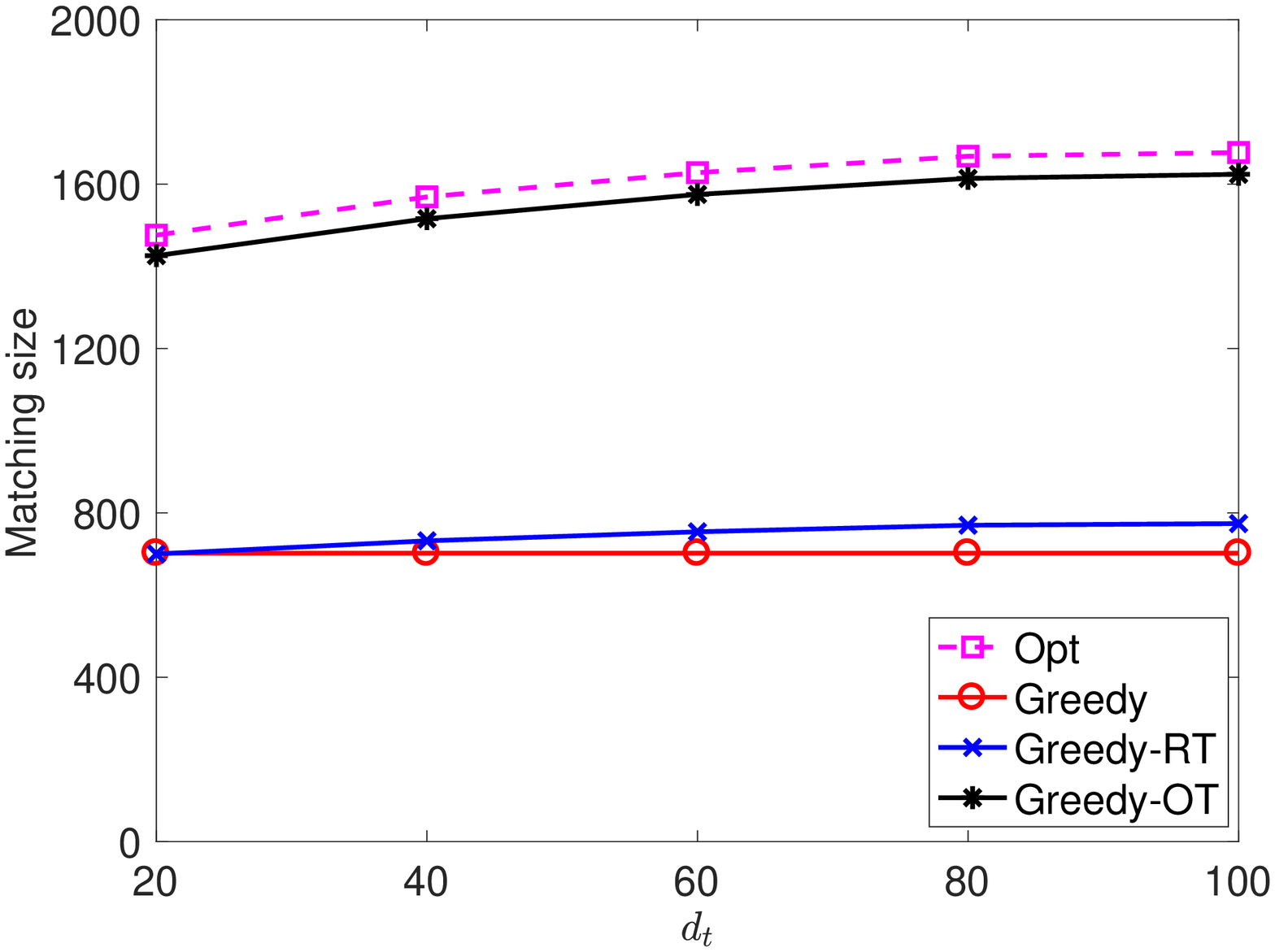}\label{fg:n_deadline}}
\vfill
\subfigure[Time of varying $|W|$]{\includegraphics[width=0.23\textwidth]{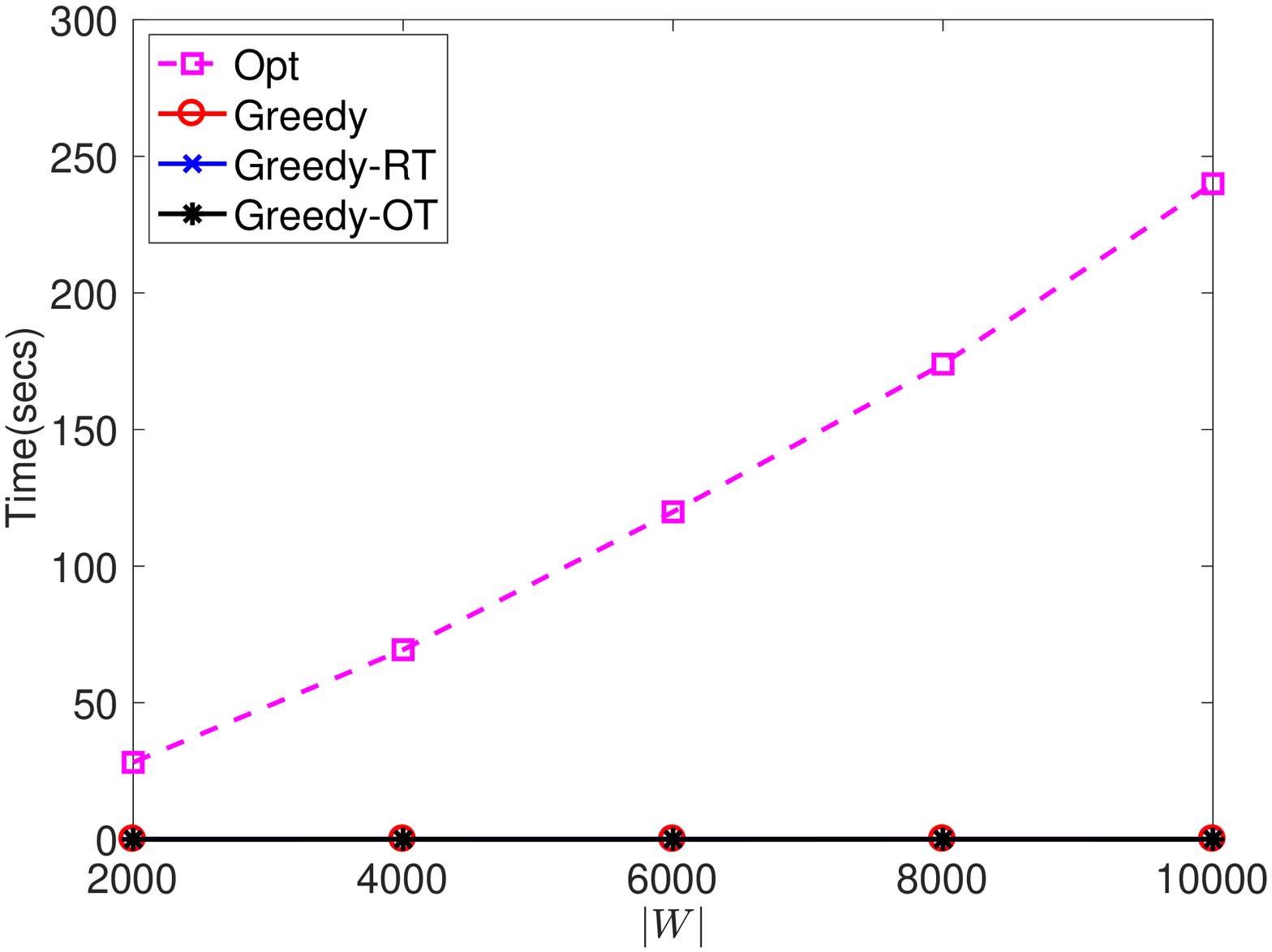}\label{fg:t_worker}}
\subfigure[Time of varying $|T|$]{\includegraphics[width=0.23\textwidth]{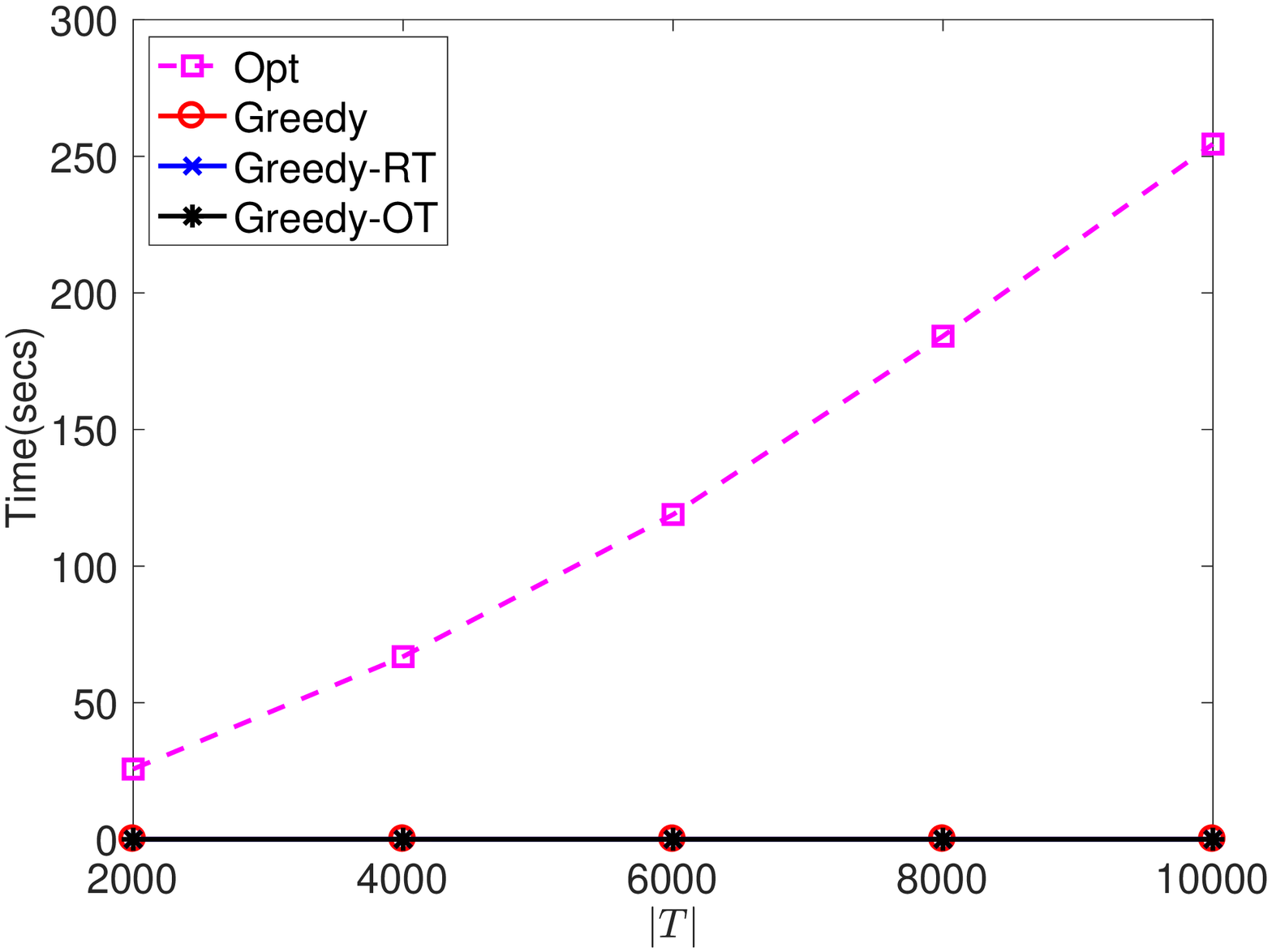}\label{fg:t_task}}
\subfigure[Time of varying $B$]{\includegraphics[width=0.23\textwidth]{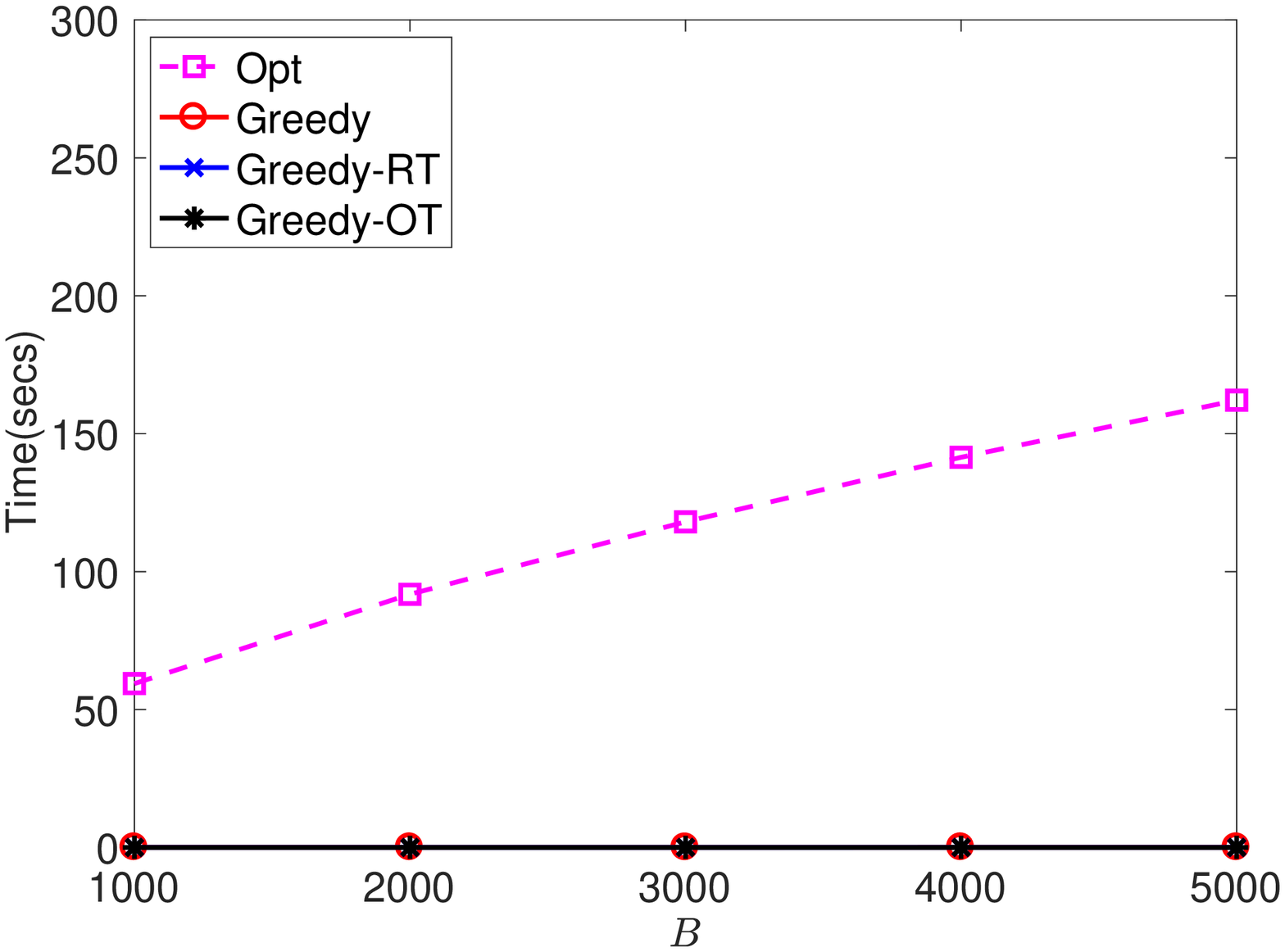}\label{fg:t_budget}}
\subfigure[Time of varying $d_t$]{\includegraphics[width=0.23\textwidth]{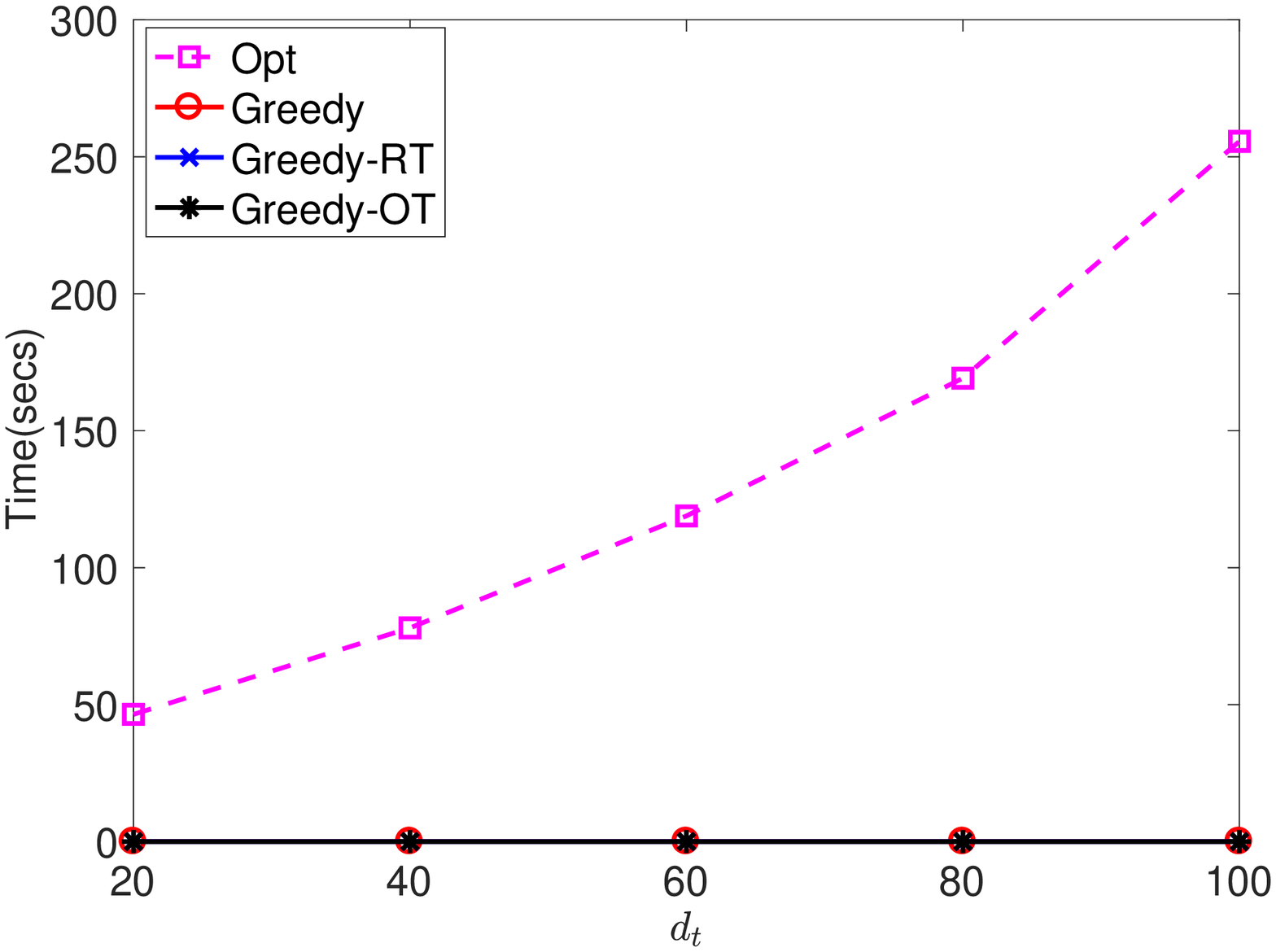}\label{fg:t_deadline}}
\vfill
\subfigure[Memory of varying $|W|$]{\includegraphics[width=0.23\textwidth]{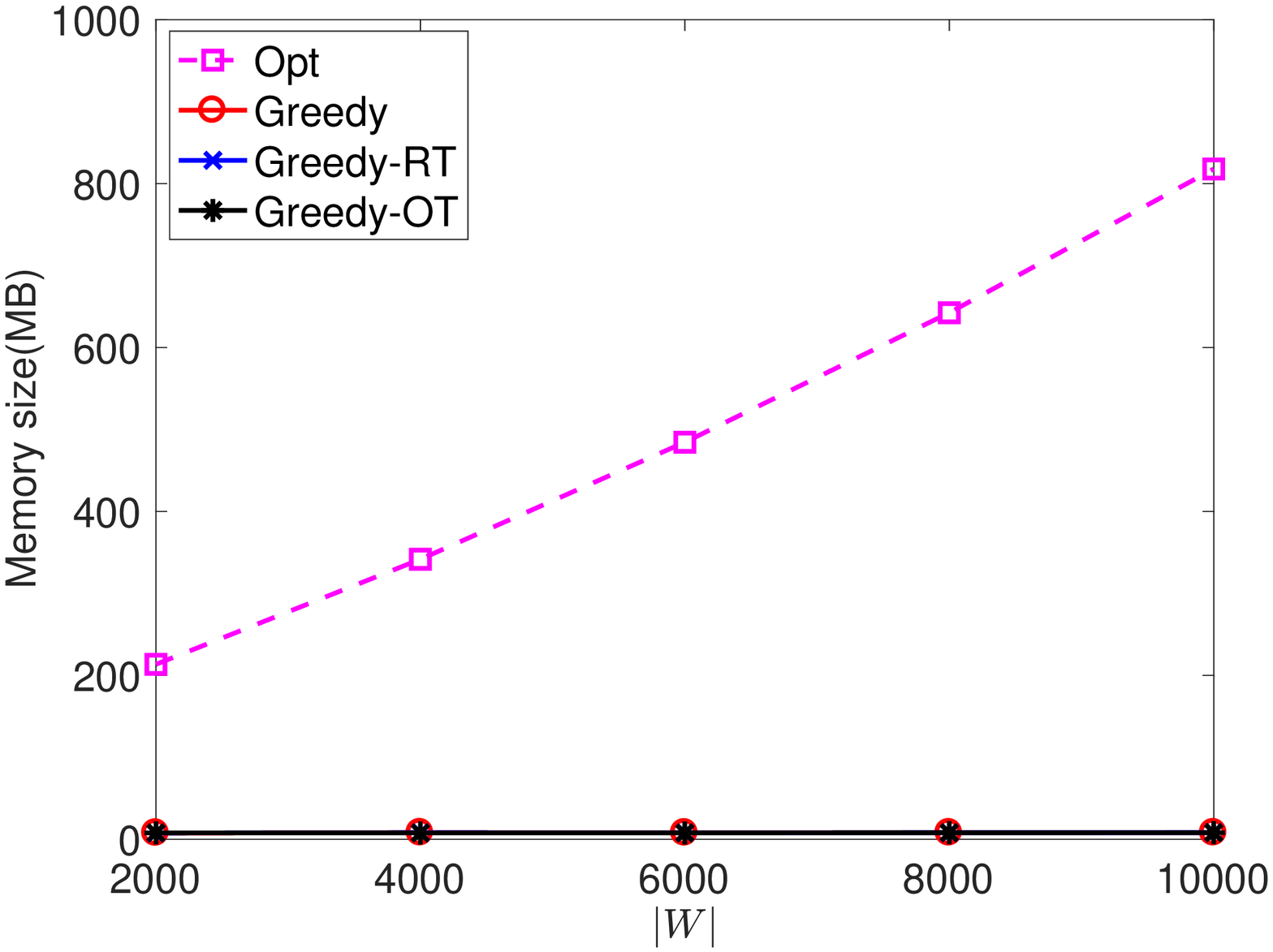}\label{fg:m_worker}}
\subfigure[Memory of varying $|T|$]{\includegraphics[width=0.23\textwidth]{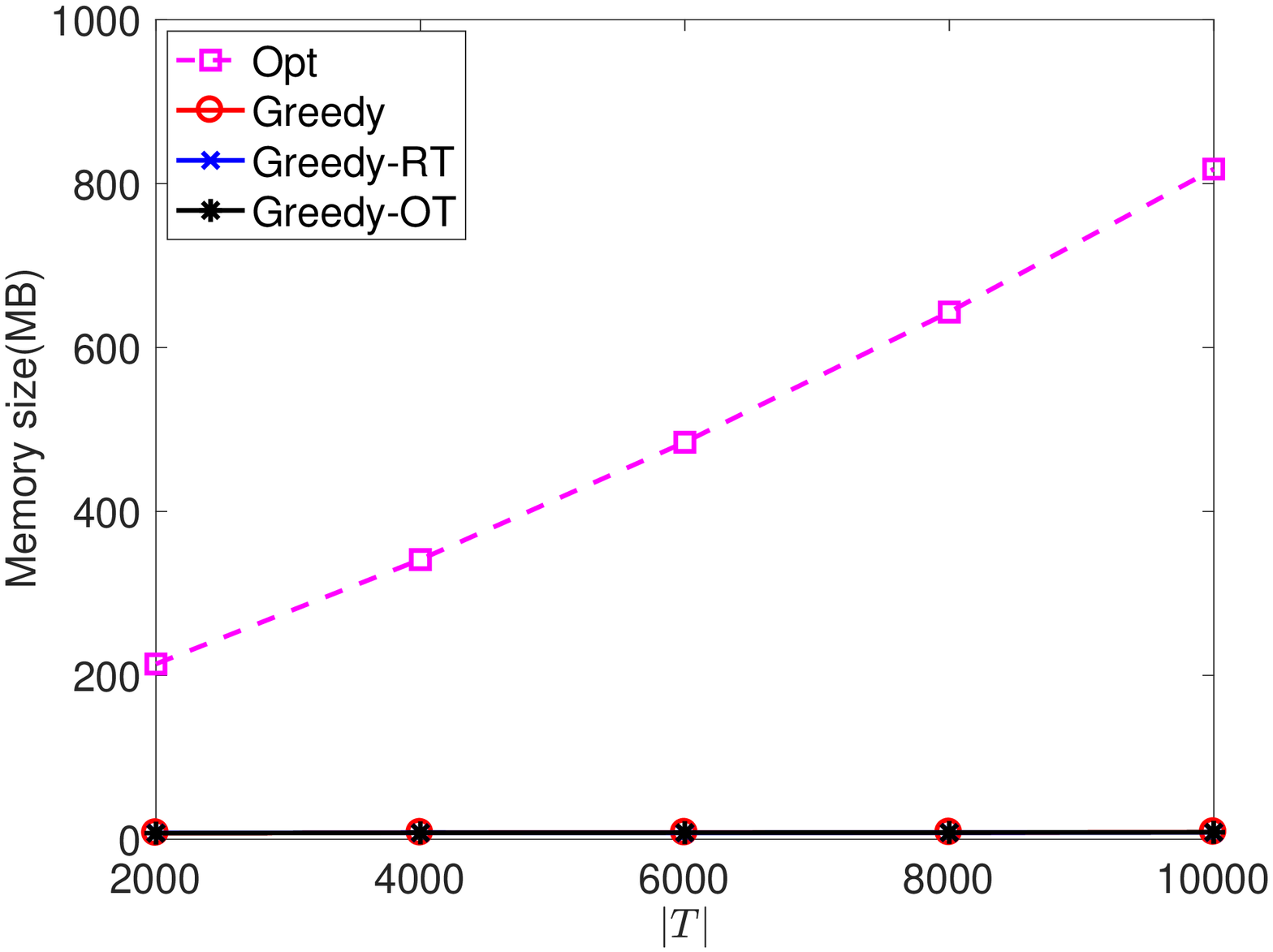}\label{fg:m_task}}
\subfigure[Memory of varying $B$]{\includegraphics[width=0.23\textwidth]{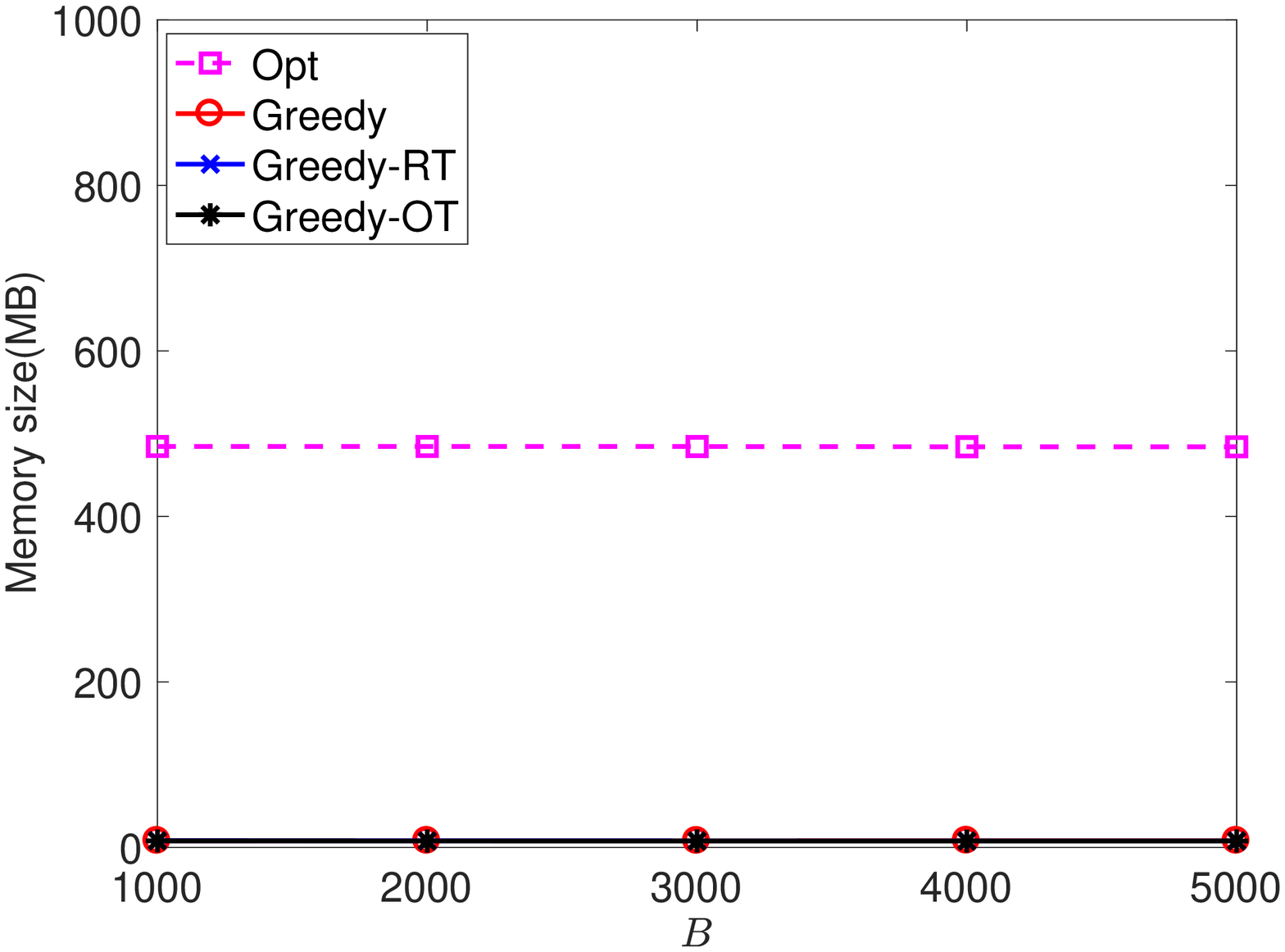}\label{fg:m_budget}}
\subfigure[Memory of varying $d_t$]{\includegraphics[width=0.23\textwidth]{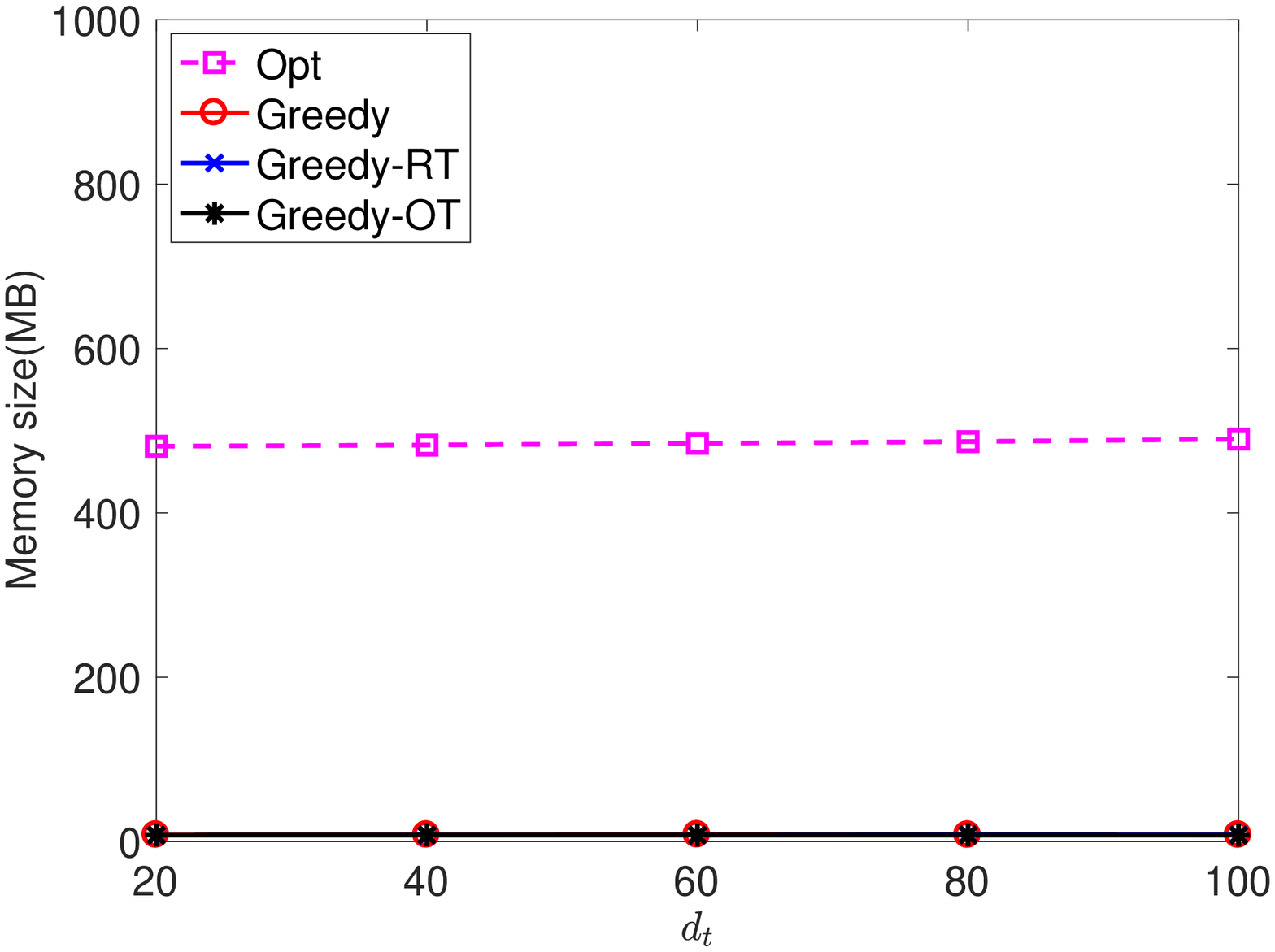}\label{fg:m_deadline}}
\caption{Variation of matching size on each specific parameter in random model }
\label{fg:rd}
\end{figure}
\subsection{Real Dataset}
We next present the experiment results on real dataset. 
As table \ref{tb:real} shows, Greedy-OT is still the best in matching size among all online algorithms, followed by Greedy-RT and simple greedy. 
Greedy-OT achieves approximately 70 percent of matching sizes compared with OPT, which is at least 1.5 times than other greedy algorithms. 
As for time and memory cost, OPT is still the most heavyweight algorithm. 
Especially in our case, the velocity of a new arrival worker is high enough that he can reach most of tasks before their deadline. Consequently, large number of edges is added when building flow graph and time cost increases dramatically. 
Compared with OPT, the other algorithms have a tiny time and memory cost.
The 2nd and 3rd column in Table \ref{tb:real} illustrate that the quantity of pickups(workers) has no significant impaction on the optimal threshold. The optimal thresholds during the whole week fluctuate less than 15 percent even if number of pickups varies in a large magnitude from 9349 to 4851, which demonstrates Greedy-OT we propose is robust. 

\begin{table}
\newcommand{\tabincell}[2]{\begin{tabular}{@{}#1@{}}#2\end{tabular}}
\caption{Matched quantity, time and memory test on Uber dataset}
\label{tb:real}
\centering
\begin{tabular}{|c|c|c|c|c|c|c|c|} \hline
Date&\tabincell{c}{No. Pickups\\(workers)}&\tabincell{c}{Opt Threshold\\(Km)}&Measures&OPT&Greedy&Greedy-RT&Greedy-OT\\\hline
\multirow{3}*{\tabincell{c}{May 5\\(Mon)}}&\multirow{3}*{5682}&\multirow{3}*{0.7956}&quantity&1100&513&569&822\\\cline{4-8}
~&~&~&time(secs)&649.254&0.088&0.067&0.045\\\cline{4-8}
~&~&~&memory(KB)&531956&8224&8010&8396\\\hline
\multirow{3}*{\tabincell{c}{May 6\\(Tue)}}&\multirow{3}*{5818}&\multirow{3}*{0.829}&quantity&1076&374&458&786\\\cline{4-8}
~&~&~&time(secs)&643.348&0.078&0.064&0.044\\\cline{4-8}
~&~&~&memory(KB)&543996&8216&8129&7908\\\hline
\multirow{3}*{\tabincell{c}{May 7\\(Wed)}}&\multirow{3}*{6065}&\multirow{3}*{\textbf{0.8356}}&quantity&1093&433&501&776\\\cline{4-8}
~&&&time(secs)&667.189&0.077&0.067&0.045\\\cline{4-8}
~&&&memory(KB)&563452&7760&8241&7872\\\hline
\multirow{3}*{\tabincell{c}{May 8\\(Thu)}}&\multirow{3}*{9349}&\multirow{3}*{0.7848}&quantity&1154&406&485&802\\\cline{4-8}
~&&&time(secs)&1086.118&0.110&0.092&0.060\\\cline{4-8}
~&&&memory(KB)&872352&8036&8054&8032\\\hline
\multirow{3}*{\tabincell{c}{May 9\\(Fri)}}&\multirow{3}*{7377}&\multirow{3}*{0.7762}&quantity&1135&385&458&763\\\cline{4-8}
~&&&time(secs)&852.985&0.091&0.074&0.044\\\cline{4-8}
~&&&memory(KB)&685856&7884&8164&7960\\\hline
\multirow{3}*{\tabincell{c}{May 10\\(Sat)}}&\multirow{3}*{4851}&\multirow{3}*{0.7362}&quantity&1196&369&458&812\\\cline{4-8}
~&&&time(secs)&711.796&0.072&0.056&0.034\\\cline{4-8}
~&&&memory(KB)&471152&8104&7882&8260\\\hline
\multirow{3}*{\tabincell{c}{May 11\\(Sun)}}&\multirow{3}*{5276}&\multirow{3}*{0.7277}&quantity&1190&346&442&802\\\cline{4-8}
~&&&time(secs)&813.825&0.067&0.054&0.016\\\cline{4-8}
~&&&memory(KB)&512360&8136&7963&8308\\\hline
\end{tabular}
\end{table}
In order to illustrate the reason why Greedy-OT can apply the historical data (May 7) to filter or process the data of other days, we draw the distribution of matching size under different travel cost from the raw optimal matching data in May 6 (the previous day), 7(current day), 8(the next day) and 14(wed of the next week). As shown in Fig. \ref{fg:trend}, all curves share a similar trend with the increase of travel cost, which dramatically rises up to their peaks and then gradually drops in the nadir. Further, We separate the whole travel cost into several segments, each of which is 0.1km, and then we can find that matching size during each cost segment is similar. For example, the average of matching size in segment [0.2,0.3] values 18.7, 19, 20.7 and 18.9 in May 6, 7, 8, 14 respectively. Thus, in the event that Greedy-OT extracts its threshold from the optimal matching scheme in May 7, it can also perform well for other days.
\begin{figure}
\centering
\caption{Similar distribution on the travel cost for different days}
\label{fg:trend}
\includegraphics[scale=0.6]{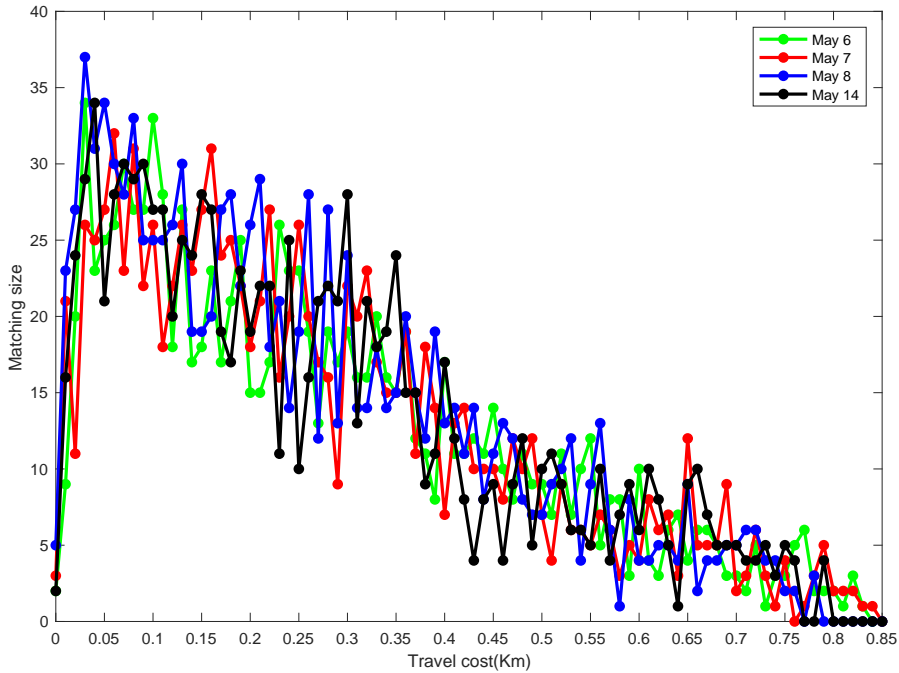}
\end{figure}

\section{Related work}
In this section, we review related works from two categories, task assignment with budget constraint and online task assignment.

\textbf{Task assignment with budget constraint}. Many studies lay stress on the budget constrained assignment problem in web crowdsourcing. 
Most of them take worker's reputation or reliability into consideration, and aim to achieve truthful and guaranteed answers.
\cite{Karger2013Budget} takes account the confidence in answers that crowd workers submit, and aims to minimize the budget to meet a certain target reliability.
\cite{Li2016Crowdsourcing} aims to maximize the number of labeled instances that achieve the quality requirement under a tight budget.
\cite{Liu2015Trust,Khetan2016Achieving,LahoutiH16} aim to find an adaptive task scheme to achieve the best trade-off between budget and accuracy of aggregated answers.
Other research focuses on finding optimal  task assignment which maximizes task completion rate under budget constraint or provides maximum profit for requesters.
\cite{Biswas2015A} improves the number of successful assignment by exploring and exploiting high-qualify workers with budgeted multi-armed bandit mechanism. 
\cite{Wu2017Budgeted} proposes linear programming based algorithms to maximizes the expected profit of assignment. 
As for the field of spatial crowdsourcing, certain studies on the budget constrained assignment problems have also been sparked.
 \cite{To2016Real,Zhao2016Budget} focus on the online maximum task coverage problem which aims to select a set of workers to maximize task coverage under a constraint budget.
\cite{Miao2016Balancing,YuMSL15} take workers' reputation and geometrical proximity into account, and aim to maximize the expected quality of  the assignment while staying within a limited budget.
\cite{Cheng2017Prediction} formulates the maximum quality task assignment problem with the optimization objective to maximize a global assignment quality score under a traveling budget constraint. 
The aforementioned works in spatial crowdsourcing stress on the worker selection and quality control, which differ from our optimal goal of maximizing matching size.

\textbf{Online task assignment}. The online task assignment problem in the field of spatial crowdsourcing has been studied recently. 
According to weather both of workers and tasks appear on platform dynamically, existing research can be divided into two categories: one-side and two-sides online assignment.  
\cite{Hassan2014A,Tong2016Online,She2016Conflict} focus on one-side online assignment. 
\cite{Tong2016Online} presents a comprehensive experimental comparison of representative algorithms \cite{Kalyanasundaram1991On,Bansal2014A,Meyerson2006Randomized} for online minimum bipartite matchings. 
\cite{Hassan2014A} maximizes the number of successful assignments as spatial tasks arrive in an online manner.
\cite{She2016Conflict} extents bipartite matching to social network, and solves the event-participant arrangement problem with conflicting and capacity constraint when users arrive on platform dynamically.
\cite{Tong2016Onlinemobil,Tong2017Flexible,Song2017Trichromatic} focus on two-sides online assignment.
\cite{Tong2016Onlinemobil} devotes to allocate micro-tasks to suitable crowd workers in online scenarios, where all the spatiotemporal information of micro tasks and crowd workers are unkown.
\cite{Tong2017Flexible} guides workers' movements based on the prediction of distribution of workers and tasks to optimize the online task assignment.
Besides the traditional bipartite online matching based on workers and tasks, \cite{Song2017Trichromatic} presents the trichromatic online matching in real-time spatial crowdsourcing which comprise three entities of worker, tasks and workplace. 
Our solution for BOA attaches to one-side online task assignment problem since platform acquires tasks' spatiotemporal information in advance while is unaware of workers' until they appear. Different from aforementioned one-side assignment, we consider the budget constraint and receive the guidance provided by historical data.

\section{Conclusion}
In this paper, we formally define a dynamic task assignment problem, called budget-constraint online task assignment (BOA) in real-time spatial crowdsourcing.
We first prove the optimal solution of BOA can be solved with min-cost max-flow algorithm, and then propose two greedy variants to solve the approximate solutions. The first variant named Greedy-RT, which has the competitive ratio of $\frac{1}{\lceil ln c_{max}+1\rceil+1}$, generates a random threshold to abandon those large cost pairs to reduce the abuse of budget. In order to improve the stability of Greedy-RT, we further propose another variant called Greedy-OT, which learns a near optimal threshold from historical spatiotemporal information of workers and achieves the competitive ratio of $\frac{\sum_{i=1}^N c_i n_i}{(c_{max}^*+\varepsilon)\cdot\sum_{i=1}^N n_i}$. Finally, we verify the effectiveness and efficiency of the proposed methods through extensive experiments on both synthetic and real datasets.
\nocite{*}
\bibliographystyle{splncs}
\bibliography{bib}
\end{document}